\documentclass[conference]{IEEEtran}
\IEEEoverridecommandlockouts

\usepackage{cite}
\usepackage{amsmath,amssymb,amsfonts}
\usepackage{amsmath,amsthm}
\usepackage{algorithmic}
\usepackage{graphicx}
\usepackage{textcomp}
\usepackage{subcaption}
\usepackage{enumerate}
\usepackage{xcolor}
\usepackage{booktabs}
\usepackage{url}
\usepackage[implicit=false]{hyperref}
\usepackage[linesnumbered,ruled,vlined]{algorithm2e}
\usepackage{balance}

\newcommand{\fullpaper}[1]{#1}
\newcommand{\icdepaper}[1]{}
\newcommand{\deleted}[1]{}
\newcommand{\added}[1]{\textcolor{red}{#1}}
\newcommand{\setlinepenalty}[1]{\linepenalty=#1}

\newtheoremstyle{MyStyle}%
  {3pt} %
  {3pt} %
  {} %
  {} %
  {\bfseries} %
  {.} %
  {.5em} %
  {} %
\theoremstyle{MyStyle}
\newtheorem{MyDefinition}{Definition}[section]
\newtheorem{MyLemma}{Lemma}[section]
\newtheorem{MyTheorem}{Theorem}[section]
\newtheorem{MyExample}{Example}[section]

\renewenvironment{proof}[1][\proofname]{{\noindent\bfseries #1. }}{\qed\vskip 1mm}

\begin{document}

\title{Time-Constrained Continuous Subgraph Matching Using Temporal Information for Filtering and Backtracking
\thanks{$^{*}$Contact author}
}

\author
    {
        \IEEEauthorblockN{Seunghwan Min$^{\dagger}$, Jihoon Jang$^{\dagger}$, Kunsoo Park$^{*\dagger}$, Dora Giammarresi$^{\ddagger}$, Giuseppe F. Italiano$^{\mathsection}$, Wook-Shin Han$^{\mathparagraph}$}
        \IEEEauthorblockA{
        $^{\dagger}$\textit{Seoul National University, South Korea} \hspace{3mm} $^{\ddagger}$\textit{Università Roma ``Tor Vergata''}\\
        $^{\mathsection}$\textit{LUISS University, Italy} \hspace{3mm} $^{\mathparagraph}$\textit{Pohang University of Science and Technology (POSTECH), South Korea}\\
        \texttt{$^{\dagger}$\{shmin,jhjang,kpark\}@theory.snu.ac.kr \hspace{3mm} $^{\ddagger}$giammarr@mat.uniroma2.it}\\
        \texttt{$^{\mathsection}$gitaliano@luiss.it \hspace{3mm} $^{\mathparagraph}$wshan@dblab.postech.ac.kr} 
        }
    }

\maketitle

\begin{abstract}
Real-time analysis of graphs containing temporal information, such as social media streams, Q\&A networks, and cyber data sources, plays an important role in various applications. Among them, detecting patterns is one of the fundamental graph analysis problems. In this paper, we study time-constrained continuous subgraph matching, which detects a pattern with a strict partial order on the edge set in real-time whenever a temporal data graph changes over time. We propose a new algorithm based on two novel techniques. First, we introduce a filtering technique called time-constrained matchable edge that uses temporal information for filtering with polynomial space. Second, we develop time-constrained pruning techniques that reduce the search space by pruning some of the parallel edges in backtracking, utilizing temporal information. Extensive experiments on real and synthetic datasets show that our approach outperforms the state-of-the-art algorithm by up to two orders of magnitude in terms of query processing time.
\end{abstract}


\vspace*{-1.5mm}
\section{Introduction}\label{sec:introduction}
Graphs are structures widely used to represent relationships between objects. For example, communication between users on social media, financial transactions between bank accounts, and computer network traffic can be modeled as graphs. In many cases, relationships between objects in real-world graph datasets contain temporal information about when they occurred. A graph containing temporal information is called a temporal graph. Social media streams \cite{leprovost2012temporal, fan2012graph}, Q\&A networks \cite{bhat2015effects}, and cyber data sources \cite{joslyn2013massive, haslhofer2016bitcoin} such as computer network traffic and financial transaction networks, are examples of temporal graphs.

There has been extensive research on the efficient analysis of temporal graphs such as graph mining \cite{sun2019tm, han2014chronos}, graph simulation \cite{song2014event, ma2020graph}, motif counting \cite{kovanen2011temporal, liu2019sampling}, subgraph matching \cite{redmond2013temporal, redmond2016subgraph, li2021subgraph}, and network clustering \cite{crawford2018cluenet}. In this paper, we model a temporal data graph as a streaming graph \cite{besta2020practice} and address continuous subgraph matching on temporal graphs called \textit{time-constrained continuous subgraph matching} \cite{li2019time}. Given a temporal data graph $G$ and a temporal query graph $q$, the time-constrained continuous subgraph matching problem is to find all matches of $q$ that are isomorphic to $q$ and satisfy the temporal order of $q$ over the streaming graph of $G$. By considering the temporal order along with the topological structure, one can more clearly represent various real-world scenarios, such as tracking the flow of money in financial transaction networks and monitoring the flow of packets in computer network traffic. Thus, time-constrained continuous subgraph matching is a fundamental problem that can play an important role in applications such as money laundering detection and cyber attack detection \cite{li2019time}.

\begin{figure}
\centering
    \includegraphics[width=0.5\linewidth]{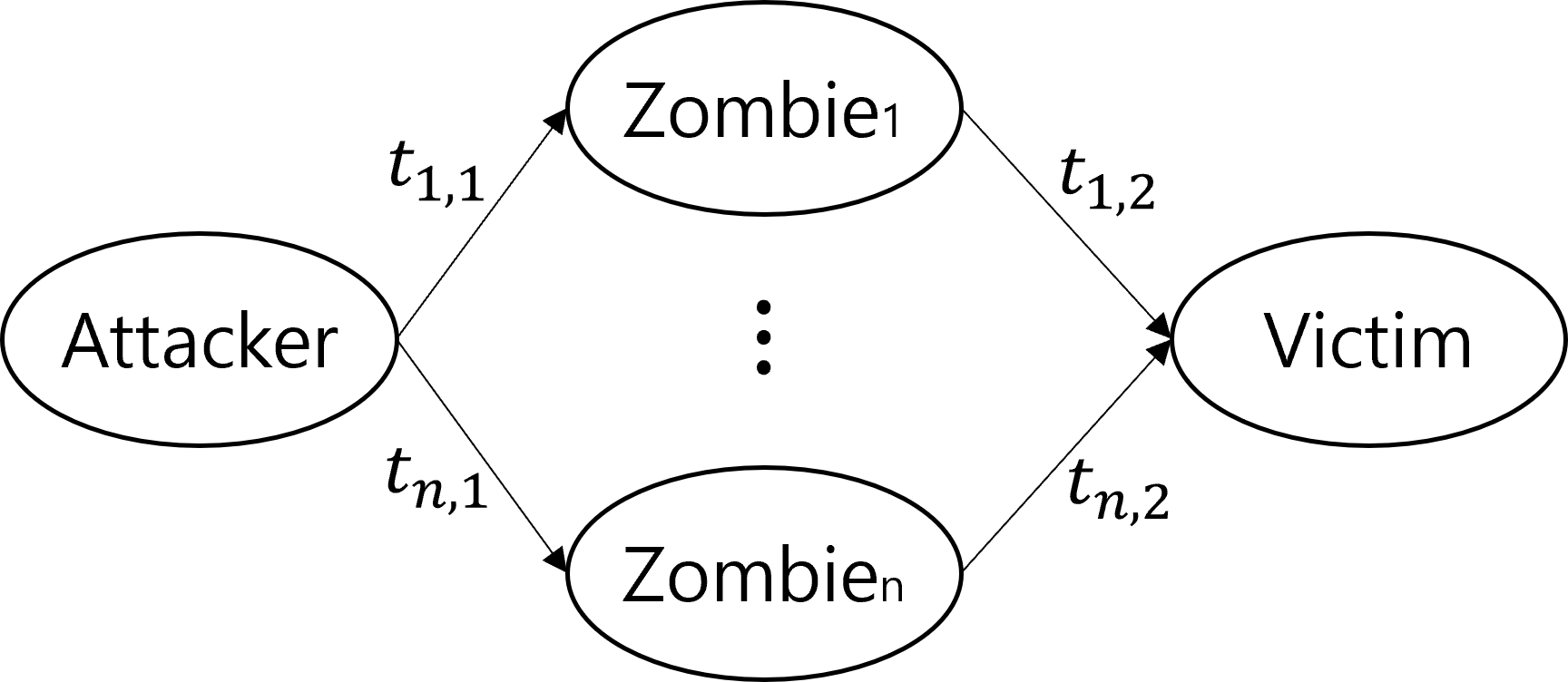}
\vspace*{-1mm}
\caption{DDoS attack pattern in network traffic}
\label{fig:ddos_attack}
\vspace*{-4mm}
\end{figure}

\begin{figure*}[t]
\centering
    \subcaptionbox{Temporal data graph $G$ where edge $\sigma_i$ arrives at time $i$\label{fig:temporal_data_graph_g}}{
        \includegraphics[width=0.22\linewidth,scale=0.2]{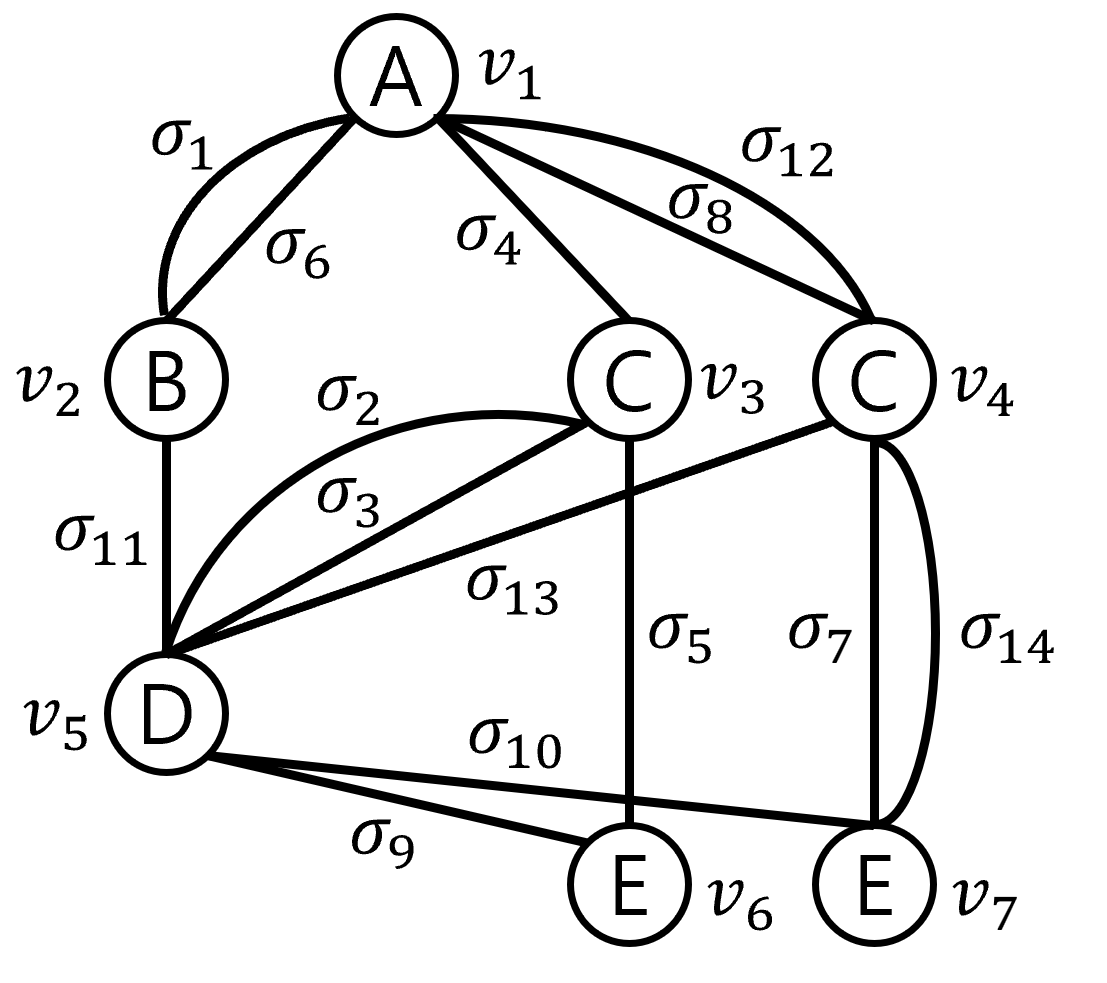}
    }
    \subcaptionbox{$G$ at $t=14$ with window $\delta=10$\label{fig:snapshot_g14}}{
        \includegraphics[width=0.22\linewidth,scale=0.2]{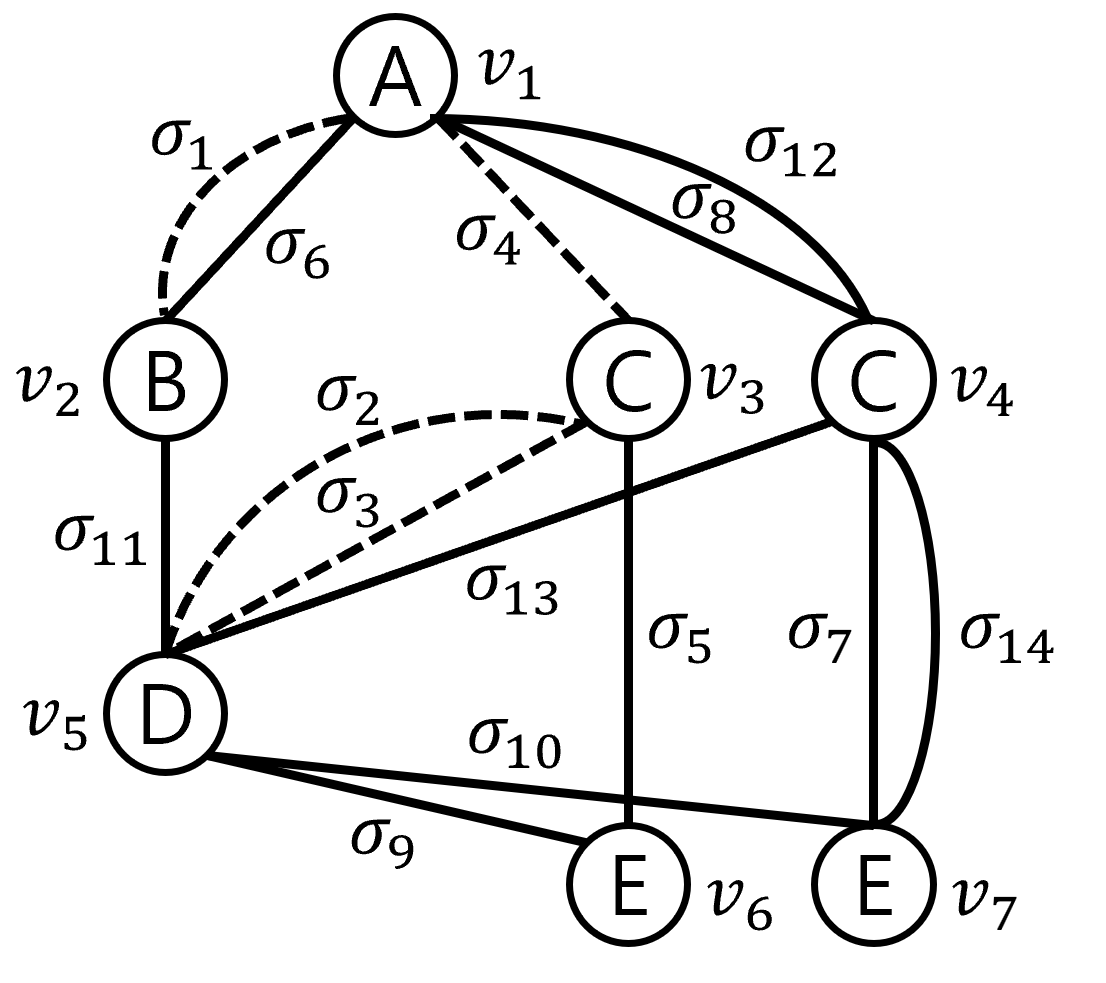}
    }
    \subcaptionbox{Query graph $q$ with temporal order constraints\label{fig:temporal_query_graph_q}}{
        \includegraphics[width=0.26\linewidth,scale=0.2]{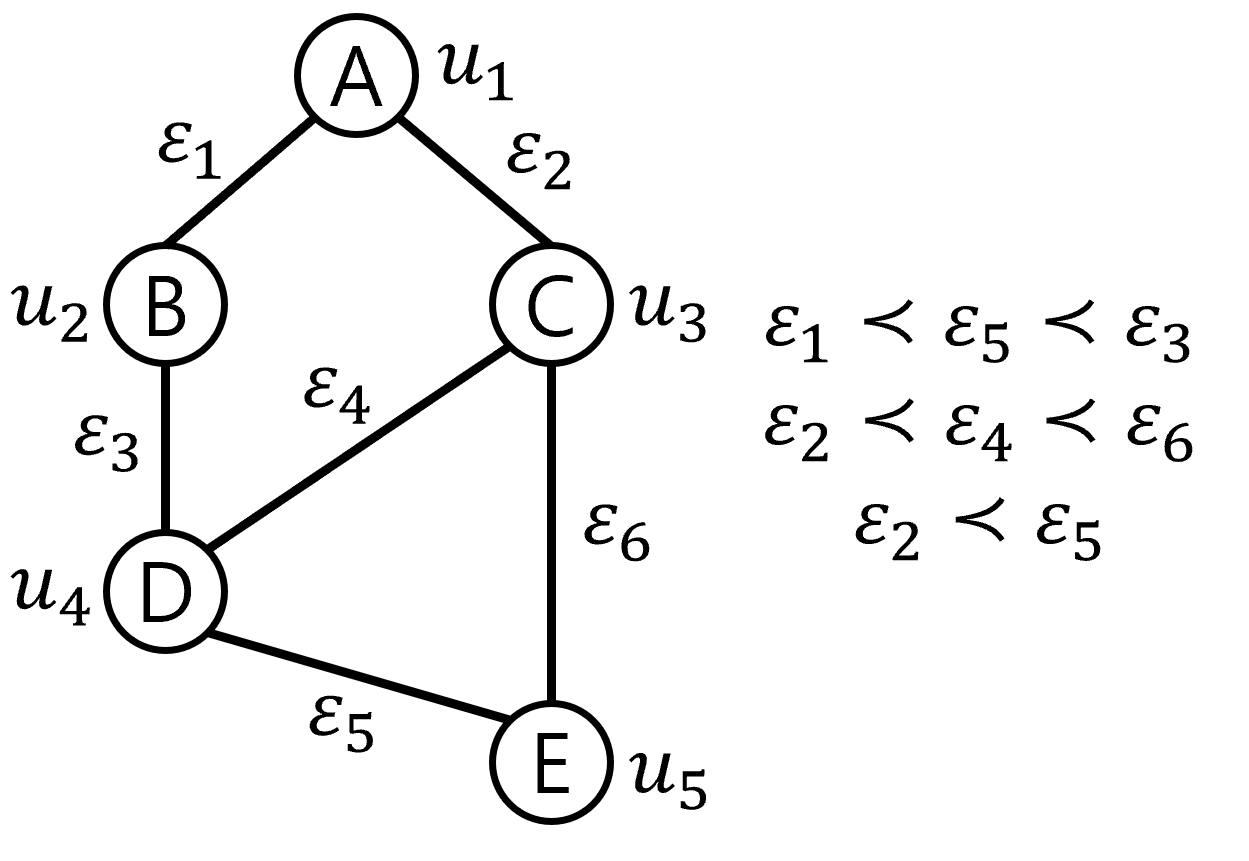}
    }
    \vspace*{-1mm}
\caption{A running example}
\label{fig:running_example}
\vspace*{-5mm}
\end{figure*}

The query depicted in Figure \ref{fig:ddos_attack} illustrates the essential pattern in DDoS attacks. Each zombie computer ($\text{Zombie}_i$) is infected through one of multiple contamination paths, which is not shown in the query graph. The zombies receive commands from the attacker ($t_{i,1}$), and then they attack the victim ($t_{i,2}$). Thus there is a temporal order $t_{i,1}\prec t_{i,2}$ for each $i$. 
Real-world DDoS attacks can be more complex than the query pattern in Figure 1, but they include the query pattern as a subgraph because the query pattern is the core of DDoS attacks. Hence, by detecting and recognizing the query pattern in network traffic data, one can identify real-world DDoS attacks that include the query pattern as a subgraph (note that one can identify the attacker by using this query pattern!). US communications company Verizon has analyzed 100,000 security incidents from the past decade that reveal that 90\% of the incidents fall into ten attack patterns \cite{verizon2020}, which can be described as graph patterns.

Many researchers have studied the continuous subgraph matching problem on non-temporal graphs \cite{fan2013incremental, kankanamge2017graphflow, choudhury2015selectivity, kim2018turboflux, min2021symmetric}. 
It is possible to solve the time-constrained continuous subgraph matching problem with these algorithms. However, these algorithms also find matches that do not satisfy the temporal order, which we don't need. Therefore, we have to remove unnecessary matches at the end, which is inefficient.
Recently, several studies have been conducted to find matches by taking the temporal order into account \cite{sun2017hasse, li2019time}. However, existing work shows limitations in that it requires exponential space to store all partial matches or utilize only temporal order constraints between incident edges.

A temporal graph typically has many parallel edges because each edge has temporal information. In the case of computer network traffic or financial transaction networks, there are many edges (i.e., data transmissions or transaction records) between two nodes. 
Multigraphs (and thus algorithms for multigraphs) are valuable in such applications due to their unique feature of parallel edges, which sets them apart from simple graphs \cite{papalexakis2013more, wang2017approximately, shafie2015multigraph}.
In this paper, we propose a new time-constrained continuous subgraph matching algorithm \texttt{TCM} for temporal graphs that adopts two novel techniques which fully exploit temporal information in both filtering and backtracking.
\begin{itemize}
    \item 
     To address the limitations of existing approaches, we propose an efficient filtering technique named \textit{time-constrained matchable edge} that can make full use of temporal relationships between non-incident edges for filtering. Our filtering approach is based on the basic idea that in order to match a query edge $e$ with a data edge $\overline{e}$, there must exist a path starting from $\overline{e}$ that satisfies the temporal order of $e$ along every possible path starting from $e$ in the query graph. Consider the process of applying the previous idea to the temporal data graph $G$ in Figure \ref{fig:temporal_data_graph_g} and the temporal query graph $q$ in Figure \ref{fig:temporal_query_graph_q}. The ID of each edge in $G$ indicates the arrival time of the edge (i.e., edge $\sigma_{i}$ arrives at time $i$). Since there is no path starting from $\sigma_{4}$ in Figure \ref{fig:temporal_data_graph_g} satisfying the temporal order constraint $\varepsilon_{2}\prec\varepsilon_{4}$ for the path $\varepsilon_{2}\rightarrow\varepsilon_{4}$ in Figure \ref{fig:temporal_query_graph_q}, we can safely exclude $\sigma_{4}$ from the matching candidates of $\varepsilon_{2}$. 
    
    We introduce a data structure called \textit{max-min timestamp} with polynomial space, which can determine whether an edge is filtered by the time-constrained matchable edge and can be updated efficiently for each graph update. For example, consider the situation of determining whether there is a path starting from $\sigma_{1}$ that satisfies the temporal order of $\varepsilon_{1}$ ($\varepsilon_{1}\prec\varepsilon_{5}$ and $\varepsilon_{1}\prec\varepsilon_{3}$) for the path of $\varepsilon_{1}\rightarrow\varepsilon_{3}\rightarrow\varepsilon_{5}$. The max-min timestamp for $\varepsilon_{1}$ at $\sigma_{1}$ stores the maximum value among the minimum timestamps of the corresponding edges following $\varepsilon_{1}$ (i.e., $\varepsilon_{3}$ and $\varepsilon_{5}$) from the various paths starting from $\sigma_{1}$.
    Since the minimum values of timestamps in two paths $\sigma_{11}\rightarrow\sigma_{9}$ and $\sigma_{11}\rightarrow\sigma_{10}$ (corresponding to the path $\varepsilon_{3}\rightarrow\varepsilon_{5}$) are 9 and 10, respectively, the max-min timestamp for $\varepsilon_{1}$ at $\sigma_{1}$ stores the maximum value of 10. We can determine that there is a path starting from $\sigma_{1}$ that satisfies the temporal order of $\varepsilon_{1}$ for the path of $\varepsilon_{1}\rightarrow\varepsilon_{3}\rightarrow\varepsilon_{5}$ because the max-min timestamp for $\varepsilon_{1}$ at $\sigma_{1}$ is greater than 1. We also introduce how to update the max-min timestamp efficiently in Section \ref{subsec:computing_time-constrained_matchable_edge}.
    
    \item 
    Second, we develop a set of time-constrained pruning techniques to reduce the search space in backtracking. Unlike non-temporal graphs, parallel edges between two vertices in a temporal graph are distinguished because they have different timestamps. We can prune some of the parallel edges by utilizing temporal relationships. 
    When finding a match for query graph $q$ (Figure \ref{fig:temporal_query_graph_q}) in the data graph $G$ (Figure \ref{fig:temporal_data_graph_g}) that satisfies the temporal order of $q$, assume that we try to match query edges $\varepsilon_{5}$, $\varepsilon_{6}$, $\varepsilon_{4}$, $\varepsilon_{3}$, $\varepsilon_{1}$, and $\varepsilon_{2}$ in this order.
    Now, consider a situation where $\varepsilon_{5}$ and $\varepsilon_{6}$ are matched to $\sigma_{9}$ and $\sigma_{5}$, respectively, and $\varepsilon_{4}$ is the next edge to be matched to either $\sigma_{2}$ or $\sigma_{3}$.
    Given that $\varepsilon_{4}\prec\varepsilon_{6}$ holds true regardless of whether $\sigma_{2}$ or $\sigma_{3}$ is chosen, the temporal order constraint that needs to be satisfied for $\varepsilon_{4}$ is reduced to $\varepsilon_{2}\prec\varepsilon_{4}$. To increase the likelihood of finding the desired match, it is more advantageous to match $\varepsilon_{4}$ to an edge with a larger timestamp. Consequently, we conduct backtracking in the direction of matching $\varepsilon_{4}$ with $\sigma_{3}$ prior to matching it with $\sigma_{2}$. The backtracking process ends without finding the desired match due to the absence of any edge that can be matched with $\varepsilon_{2}$. When we retrace our steps to the point at which we have to match $\varepsilon_{4}$ to $\sigma_{2}$ instead of $\sigma_{3}$, we can deduce from the previously unsuccessful backtracking attempt that any further backtracking will not yield the desired match. Therefore, we prune the search space that matches $\varepsilon_{4}$ to $\sigma_{2}$ and explore other search spaces.
    
    We classify the temporal relationship between the query edge $e$ and other query edges into three cases: no temporally related edges, all edges with the same temporal relationship, and the remaining case.
    For each case, we demonstrate a situation where a match cannot exist and use it to prune edges.
\end{itemize}

We conduct experiments on six real and synthetic datasets comparing our algorithm with existing algorithms. Experimental results show that \texttt{TCM} outperforms existing approaches by up to two orders of magnitude in terms of query processing time.

The rest of the paper is organized as follows. Section \ref{sec:preliminaries} defines the problem statement and provides related work. Section \ref{sec:overview} gives an overview of our algorithm. Section \ref{sec:filtering_by_temporal_order_constraint} introduces our filtering technique and Section \ref{sec:backtracking} presents techniques to prune out a part of search space. Section \ref{sec:performance} shows the results of our performance evaluation. Finally, Section \ref{sec:conclusion} concludes the paper. \icdepaper{Proofs and an additional example in Section \ref{sec:backtracking} are provided in the full version \cite{fullversion}.}

\section{Preliminaries}\label{sec:preliminaries}

In this paper, we focus on undirected and vertex-labeled graphs. Our techniques can be easily extended to directed graphs with multiple labels on vertices or edges.

\begin{MyDefinition}
\cite{paranjape2017motifs, li2021subgraph} A \textit{temporal graph} $G=(V(G), \allowbreak E(G),  L_G, T_G)$ consists of a set $V(G)$ of vertices, a set $E(G)$ of edges, a labeling function $L_G:V(G)\rightarrow \Sigma$ that assigns a label to each vertex from the set $\Sigma$ of labels, and a timing function $T_G:E(G)\rightarrow \mathbb{N}$ that assigns a timestamp to each edge. The timestamp of an edge indicates when the edge arrived. We represent timestamps as natural numbers. 
\end{MyDefinition}

A temporal graph can have parallel edges with different time\-stamps between two vertices. To distinguish edges with different timestamps between vertices $u$ and $v$, we denote each edge as $(u,v,t)$, where $t$ is the timestamp of the edge. Note that there are two edges $\sigma_1=(v_1,v_2,1)$ and $\sigma_6=(v_1,v_2,6)$ between $v_1$ and $v_2$ in Figure \ref{fig:temporal_data_graph_g}.

\begin{MyDefinition}
A \textit{temporal query graph} $q$ is defined as $(V(q), E(q), \allowbreak L_q, \prec )$ where $(V(q), E(q), L_q)$ is a graph and $\prec$ is a strict partial order relation on $E(q)$ called the \textit{temporal order}. We say that $e_1$ and $e_2$ are temporally related when $e_1\prec e_2$ or $e_2\prec e_1$.
\end{MyDefinition}

\begin{MyDefinition}
Given a temporal query graph $q=(V(q), \allowbreak E(q), L_q,  \prec )$ and a temporal data graph $G=(V(G), \allowbreak E(G), L_G, T_G)$, a \textit{time-constrained embedding} of $q$ in $G$ is a mapping $M:V(q)\cup E(q)\rightarrow V(G)\cup E(G)$ such that (1) $M$ is injective, (2) $L_q(u)=L_G(M(u))$ for every $u\in V(q)$, (3) $M(e)\in E(G)$ is an edge between $M(u)$ and $M(u')$ for every $e=(u,u')\in E(q)$ (i.e., $M(e)=(M(u),M(u'),t)$), and (4) $e_1 \prec e_2 \Rightarrow T_G(M(e_1))\allowbreak < T_G(M(e_2))$ for any two $e_1 , e_2\in E(q)$.
\end{MyDefinition}

A mapping that satisfies (2) and (3) is called a \textit{homomorphism}. A homomorphism satisfying (1) is called an \textit{embedding}. 
Homomorphism and embedding are widely used in graph matching problems. In this paper, we use embeddings as matching semantics.

\begin{MyExample}
In a temporal data graph $G$ in Figure \ref{fig:temporal_data_graph_g} and a temporal query graph $q$ in Figure \ref{fig:temporal_query_graph_q}, two mappings $\{(\varepsilon_1,\sigma_1),(\varepsilon_2,\sigma_8),(\varepsilon_3,\sigma_{11}),(\varepsilon_4,\sigma_{13}),(\varepsilon_5,\sigma_{10}),(\varepsilon_6,\sigma_{14})\}$~and $\{(\varepsilon_1,\sigma_6),(\varepsilon_2,\sigma_8),(\varepsilon_3,\sigma_{11}),(\varepsilon_4,\sigma_{13}),(\varepsilon_5,\sigma_{10}),(\varepsilon_6,\sigma_{14})\}$~are embeddings and also time-constrained embeddings of $q$ in $G$. 
In contrast, a mapping $M=\{(\varepsilon_1,\sigma_1),(\varepsilon_2,\sigma_4), (\varepsilon_3,\sigma_{11}),\allowbreak (\varepsilon_4,\sigma_{2}),(\varepsilon_5,\sigma_{9}),(\varepsilon_6,\sigma_{5})\}$ is an embedding of $q$ in $G$, but it is not a time-constrained embedding because $\varepsilon _2 \prec \varepsilon _4$ in $q$ but $T_{G}(M(\varepsilon_2))=T_{G}(\sigma_4)=4 \nless T_{G}(M(\varepsilon_4))=T_{G}(\sigma_2)=2$ in $G$. When representing the mapping $M$, we omit a pair of a vertex $u\in V(q)$ and its image $M(u)$ for simplicity of notation.
\end{MyExample}

In order not to find time-constrained embeddings where the time interval between the minimum and maximum timestamps of its edges is too long, we model a temporal graph as a streaming graph with a time window as in \cite{li2019time}.
For a given time window $\delta$ and current time $t$, the edges with timestamp less than or equal to $t-\delta$ expire. In other words, we keep only the edges that arrive between $(t-\delta,t]$ and find time-constrained embeddings over these edges.

\begin{MyExample}
Figure \ref{fig:snapshot_g14} shows the temporal graph at time $t=14$ with the time window $\delta=10$. At time $t=14$, the edge $\sigma_{14}$ arrives and the edge $\sigma_4$ expires. When $\sigma_{14}$ arrives, a time-constrained embedding $\{(\varepsilon_1,\sigma_6),(\varepsilon_2,\sigma_8),(\varepsilon_3,\sigma_{11}),(\varepsilon_4,\allowbreak\sigma_{13}),(\varepsilon_5,\sigma_{10}),(\varepsilon_6,\sigma_{14})\}$ occurs. In this case, since the edge $\sigma_1$ has already expired, a time-constrained embedding $\{(\varepsilon_1,\sigma_1),(\varepsilon_2,\sigma_8),(\varepsilon_3,\allowbreak\sigma_{11}),(\varepsilon_4,\allowbreak\sigma_{13}),(\varepsilon_5,\sigma_{10}),\allowbreak(\varepsilon_6,\sigma_{14})\}$ does not occur. At time $t=16$, the edge $\sigma_6$ expires and the time-constrained embedding containing $\sigma_6$ also expires.
\end{MyExample}

\noindent\textbf{Problem Statement.} Given a temporal data graph $G$, a temporal query graph $q$, and a time window $\delta$, the \textit{time-constrained continuous subgraph matching problem} is to find all time-constrained embeddings of $q$ that occur or expire with the time window $\delta$ according to the arrival or expiration of each edge in $G$.

\begin{MyTheorem}
\cite{li2019time} The time-constrained continuous subgraph matching problem is NP-hard.
\end{MyTheorem}

One naive solution to this problem is to solve the continuous subgraph matching (i.e., find embeddings without considering the temporal order) and then exclude the embeddings that do not satisfy the temporal order to obtain the time-constrained embeddings. In general, it would be better to filter the edges in consideration of the temporal order or prune the partial embeddings that do not fit the constraint rather than the naive solution. 

\begin{figure}[t]
\centering
    \subcaptionbox{Query DAG $\hat{q}$\label{fig:DAG_q}}{
        \includegraphics[width=0.26\linewidth]{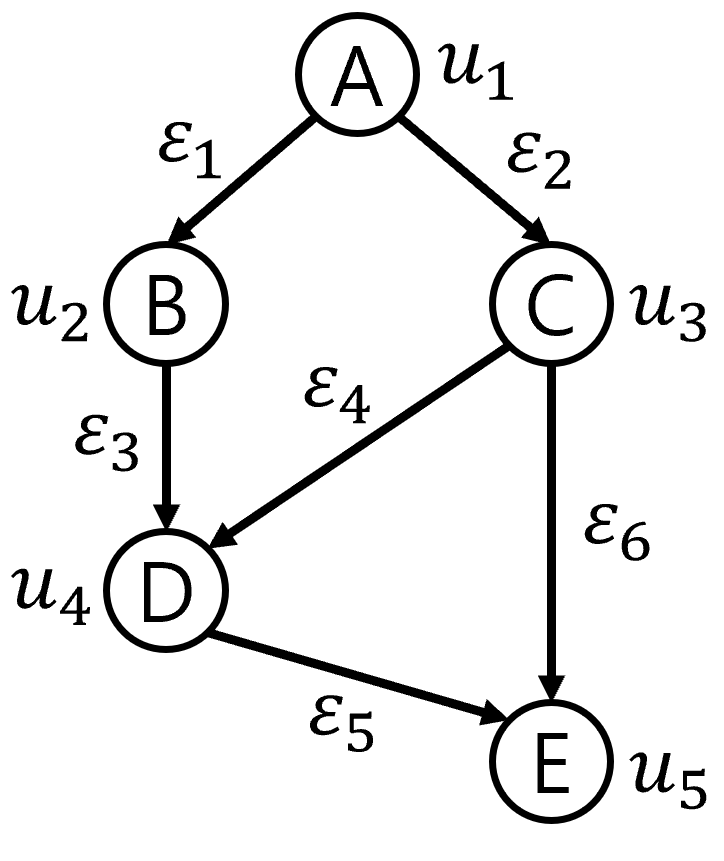}
    }
    \subcaptionbox{Query DAG $\hat{q}^{-1}$\label{fig:DAG_q-1}}{
        \includegraphics[width=0.26\linewidth]{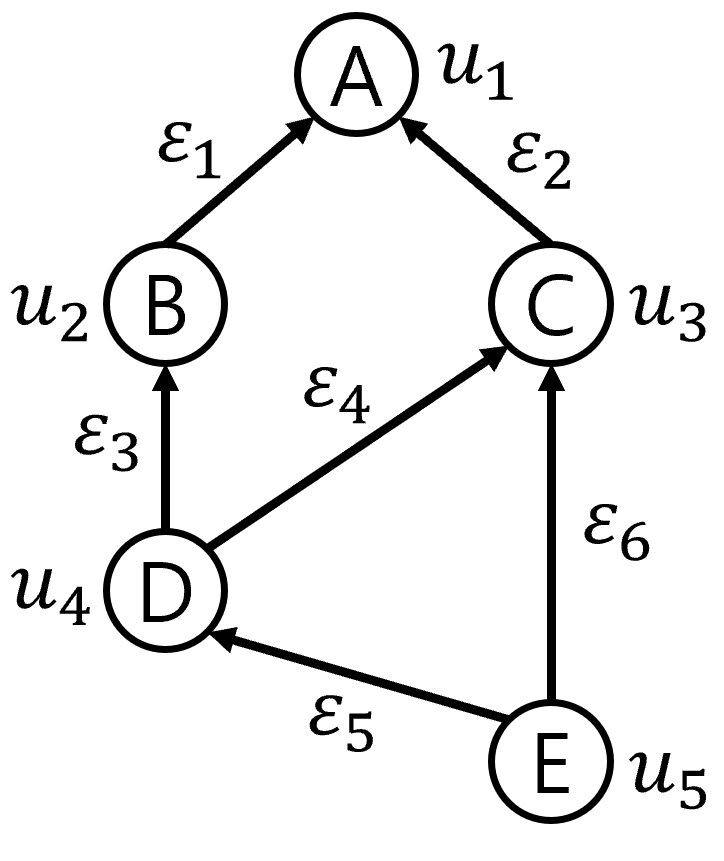}
    }
    \subcaptionbox{Path tree of $\hat{q}$\label{fig:path_tree}}{
        \includegraphics[width=0.30\linewidth]{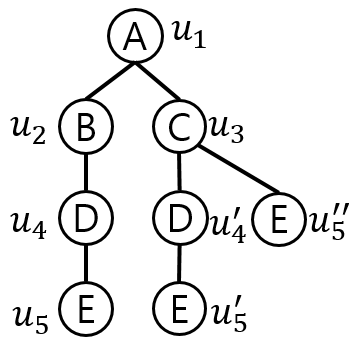}
    }
    \vspace*{-1mm}
\caption{{DAGs of $q$ in Figure \ref{fig:temporal_query_graph_q} for the running example}}
\label{fig:DAG_path_tree}
\vspace*{-6mm}
\end{figure}

In this paper, we will use a directed acyclic graph (DAG) and a weak embedding \cite{han2019efficient} as a tool to filter edges using temporal order (Section \ref{sec:filtering_by_temporal_order_constraint}). A \textit{directed acyclic graph} (DAG) $\hat{q}$ is a directed graph that contains no cycles. A \textit{root} (resp., \textit{leaf}) of a DAG is a vertex with no incoming (resp., outgoing) edges. A DAG $\hat{q}$ is a \textit{rooted DAG} if there is only one root. The DAG $\hat{q}$ in Figure \ref{fig:DAG_q} is one of the rooted DAGs that can be obtained from $q$ in Figure \ref{fig:temporal_query_graph_q}.
Its reverse $\hat{q}^{-1}$ in Figure \ref{fig:DAG_q-1} is the same as $\hat{q}$ except that all of the edges are reversed. 
\deleted{The \textit{height} of a rooted DAG $\hat{q}$ is the maximum distance between the root and any other vertex in $\hat{q}$, where the distance between two vertices is the number of edges in a shortest path connecting them.}

In a DAG, we say that $u$ is a \textit{parent} of $v$ ($v$ is a \textit{child} of $u$) if there exists a directed edge from $u$ to $v$. An \textit{ancestor} of a vertex $v$ is a vertex which is either a parent of $v$ or an ancestor of a parent of $v$. A \textit{descendant} of a vertex $v$ is a vertex which is either a child of $v$ or a descendant of a child of $v$. When representing an edge of a DAG as a pair $(u_1,u_2)$ of vertices, we will use the convention that the first vertex $u_1$ is always the parent of the second vertex $u_2$.
For edges, we say that $e_1=(u_1,u_1')$ is an ancestor of $e_2=(u_2,u_2')$ if $u_1'=u_2$ or $u_1'$ is an ancestor of $u_2$. For example, $\varepsilon_2$ is an ancestor of $\varepsilon_4$, $\varepsilon_5$, and $\varepsilon_6$ in Figure \ref{fig:DAG_q}.
Let $\textrm{Child}(u)$ and $\textrm{Parent}(u)$ denote the children and parents of $u$ in $\hat{q}$, respectively.

\begin{MyDefinition}\label{def:timing_ancestor}
Given a temporal query graph $q$ and its DAG $\hat{q}$, $e_1$ is a \textit{temporal ancestor} of $e_2$ in $\hat{q}$, denoted by $e_1\leadsto_{\hat{q}}e_2$ if $e_1$ is an ancestor of $e_2$ in $\hat{q}$ and two edges have a temporal relation (i.e., $e_1\prec e_2$ or $e_2\prec e_1$). We say that such $e_2$ is a \textit{temporal descendant} of $e_1$. 
\end{MyDefinition}

For simplicity of presentation, when $e_1$ is a temporal ancestor of $e_2$ in $\hat{q}$, we will only consider the case $e_1 \prec e_2$ (i.e., we omit the explanation for the case $e_2 \prec e_1$). The other case (i.e., $e_2 \prec e_1$) has been implemented in a symmetrical way.

\begin{MyDefinition}\label{def:sub-DAG}
A \textit{sub-DAG of $\hat{q}$ starting at $u$ or} $e$, denoted by $\hat{q}_{u}$ or $\hat{q}_{e}$, is the subgraph of $\hat{q}$ consisting of edges of all paths starting at $u$ or $e$. For example, $\hat{q}_{u_3}$ in Figure~\ref{fig:DAG_q} is the subgraph that contains $\varepsilon_4$, $\varepsilon_5$ and $\varepsilon_6$, and $\hat{q}_{\varepsilon_2}$ in Figure \ref{fig:DAG_q} consists of $\varepsilon_2,\varepsilon_4,\varepsilon_5$ and $\varepsilon_6$. 
\end{MyDefinition}

\begin{MyDefinition}\label{def:path_tree}
\cite{han2019efficient} The \textit{path tree} of a rooted DAG $\hat{q}$ is defined as the tree $\hat{q}^T$ such that each root-to-leaf path in $\hat{q}^T$ corresponds to a distinct root-to-leaf path in $\hat{q}$, and $\hat{q}^T$ shares common prefixes of its root-to-leaf paths. Figure \ref{fig:path_tree} shows the path tree of $\hat{q}$ in Figure \ref{fig:DAG_q}. 
\end{MyDefinition}

\begin{MyDefinition}\label{def:weak_embedding}
\cite{han2019efficient} A \textit{weak embedding} $M'$ of $\hat{q}_u$ at $v\in V(G)$ is defined as a homomorphism of the path tree of $\hat{q}_{u}$ in $G$ such that $M'(u)=v$. Similarly, a weak embedding $M'$ of $\hat{q}_e$ at $\overline{e}\in E(G)$ is defined as a homomorphism of the path tree of $\hat{q}_e$ in $G$ such that $M'(e)=\overline{e}$.
\end{MyDefinition}

Table \ref{tab:notations} lists the notations frequently used in this paper.

\begin{table}[h]
    \centering
    \caption{Frequently used notations}
    \vspace*{-1mm}
    \begin{tabular}{cl}
        \toprule
        Symbol & \multicolumn{1}{c}{Description} \\
        \midrule
        $G$ & Temporal data graph \\
        $q$ & Temporal query graph \\
        $\hat{q}$ & Query DAG \\
        $M(u)$ and $M(e)$ & Mapping of $u$ and $e$ in (partial) embedding $M$\\
        $\mathcal{T}(\hat{q})$ & Max-min timestamps of $\hat{q}$\\
        $EC_M(e)$ & Set of candidate edges of $e$ regarding a partial\\ &  time-constrained embedding $M$\\
        \bottomrule
    \end{tabular}
    \label{tab:notations}
    \vspace{-4mm}
\end{table}

\subsection{Related Work}\label{subsec:related_work}

\noindent\textbf{Subgraph matching.}
Subgraph matching finds all embeddings of a query graph over the static data graph. Ullmann~\cite{ullmann1976algorithm} first proposes a backtracking algorithm to find all embeddings of the query graph. Extensive studies have been done to improve the backtracking algorithm, such as \texttt{Turbo\textsubscript{iso}}~\cite{han2013turboiso}, \texttt{CFL-Match}~\cite{bi2016efficient}, \texttt{DAF}~\cite{han2019efficient} and \texttt{VEQ}~\cite{kim2021versatile}. \deleted{\texttt{Turbo\textsubscript{iso}}\deleted{~\cite{han2013turboiso}} and \texttt{CFL-Match}\deleted{~\cite{bi2016efficient}} use a spanning tree of $q$ to decide effective matching order and filter out some candidate vertices that cannot match query vertices. Instead of using the spanning tree of $q$, \texttt{DAF}\deleted{~\cite{han2019efficient}} \added{ and \texttt{VEQ}} uses a query DAG $\hat{q}$ of $q$ to achieve better filtering power and matching order. \deleted{It builds an auxiliary data structure called \textit{candidate space}, or \texttt{CS}, to filter the candidate vertices by conducting the dynamic programming between $\hat{q}$ and $G$. \texttt{VEQ}~\cite{kim2021versatile} extends \texttt{DAF} by adding a filtering technique called \textit{neighbor-safety} and improving the matching order by considering \textit{static equivalence} of query vertices. Additionally, it prunes some symmetric search space by considering \textit{dynamic equivalence} on \texttt{CS} while backtracking.}}

\deleted{\noindent\textbf{Continuous subgraph matching. }}
\deleted{Given a query graph $q$, an initial data graph $G_0$, and a graph update stream $\Delta G$, continuous subgraph matching finds all newly occurred or expired embeddings of $q$ for each inserted or deleted edge of $\Delta G$.} 
Continuous subgraph matching finds all newly occurred or expired embeddings of the query graph over the dynamic data graph. There also has been extensive studies on continuous subgraph matching problem, such as \texttt{IncIsoMatch}~\cite{fan2013incremental}, 
 \texttt{Graphflow}~\cite{kankanamge2017graphflow}, \texttt{SJ-tree}~\cite{choudhury2015selectivity}, \texttt{TurboFlux}~\cite{kim2018turboflux}, \texttt{SymBi}~\cite{min2021symmetric} and \texttt{RapidFlow} \cite{sun2022rapidflow}.
\deleted{There also has been extensive studies on continuous subgraph matching problem, which finds all newly occurred or expired embeddings of the query graph over the dynamic data graph. }
\deleted{
\texttt{IncIsoMatch}~\cite{fan2013incremental} and \texttt{Graphflow}~\cite{kankanamge2017graphflow} computes the updated embeddings without index structures updated incrementally. \texttt{SJ-tree}~\cite{choudhury2015selectivity} uses an index structure called \texttt{SJ-tree}, where the root node stores all embeddings of $q$ while the other nodes store partial embeddings. Since its space usage can be exponential,} 
\texttt{TurboFlux}\deleted{~\cite{kim2018turboflux}} proposes an auxiliary data structure called \textit{data-centric graph}, or \texttt{DCG}, which uses a spanning tree of $q$ to filter out candidate edges that can match each edge of the spanning tree. \texttt{SymBi}\deleted{~\cite{min2021symmetric}} proposes an auxiliary data structure called \textit{dynamic candidate space}, or \texttt{DCS}, which uses a query DAG of $q$ instead of the spanning tree of $q$. Since it uses a query DAG, it can filter out more candidate edges by considering non-tree edges of the spanning tree of $q$, while achieving polynomial space. \texttt{RapidFlow}\deleted{\cite{sun2022rapidflow}} addresses and solves the following two issues in existing works: matching orders always starting from the inserted edge can lead to inefficiency, and automorphism on query graph causes duplicate computation. \deleted{An in-depth study \cite{sun2022depth} investigates various recent continuous subgraph matching algorithms and shows which algorithm's index and matching order have strengths according to the characteristics of a data graph and a query graph through extensive experiments.}

\noindent\textbf{Time-constrained subgraph matching.}
There are several studies to solve subgraph matching problem on temporal graphs, such as \texttt{TOM}\cite{li2021subgraph}, \texttt{Timing}\cite{li2019time} and \texttt{Hasse} \cite{sun2017hasse}. \par

\texttt{TOM} finds all the time-constrained embeddings within a time window over the static temporal graph. It builds an index structure called \textit{Temporal-order tree} (TO-tree for short) to filter the candidate edges that can be matched to the query edge, by checking adjacency and temporal order between matching candidates. After the TO-tree construction, it enumerates all embeddings of the query graph in an edge-by-edge expansion manner. It does not fully utilize the temporal order constraints while filtering because it does not consider the constraints between non-incident edges. 
\par

\texttt{Timing} and \texttt{Hasse} solves the time-constrained continuous subgraph matching problem. They decompose a query graph into the subqueries, and store all partial embeddings of each subquery in its index structure.
When a partial embedding occurs or expires due to the edge arrival or expiration, they compute new or expired embeddings of the query graph by joining the partial embeddings. Since they materialize the partial embeddings, they may require exponential space to the size of the query graph.

\deleted{
\texttt{Hasse} also solves time-constrained continuous subgraph matching. It decomposes a query graph into subqueries called \textit{Hasse-diagram based graph}, in which a node and an edge of the graph represent a query edge and a direct partial order between query edges, respectively. It stores all partial embeddings of each subquery in an index structure called \textit{Hasse-cache structure}.
When a new edge arrives, it updates the partial embeddings. If new partial embeddings occur, it computes the embeddings of the query graph using the partial embeddings. Like \texttt{Timing}, its index structure may require exponential space.}

\section{Overview of our algorithm}\label{sec:overview}

Algorithm \ref{alg:overview} shows the overview of our time-constrained continuous subgraph matching algorithm. Given a temporal data graph $G$, a temporal query graph $q$, and a time window $\delta$, our algorithm consists of the following steps.

    \setlinepenalty{200}
    
    \noindent\textbf{Building a query DAG.} We first build a rooted DAG $\hat{q}$ from $q$ by assigning directions to the edges in $q$. The built query DAG $\hat{q}$ is used to filter edges according to the definition of ``time-constrained matchable edge'' (Section \ref{subsec:time-constrained_matchable_edge}). Our filtering process considers each ordered pair of edges which are in the temporal ancestor-descendant relationship in the query DAG. Therefore, we build a query DAG through a greedy algorithm so that many ordered pairs can be considered for filtering. Procedure \texttt{BuildDAG} takes a query graph $q$ and a vertex $r\in V(q)$ as inputs, and outputs a query DAG $\hat{q}_r$ (whose root is $r$, and that is obtained through the greedy algorithm) and its score $S_r$ (Section \ref{subsec:building_query_DAG}). Here, the score $S_r$ of DAG $\hat{q}_r$ is the number of the ordered pairs of edges in the temporal ancestor-descendant relationship in $\hat{q}_r$. We try every vertex once as the root and use the query DAG with the highest score (Lines 1--6). 
    
    \setlinepenalty{10}
    
    $L$ is a set that contains events of the arrival and expiration of edges in $G$. Events for edge arrival are denoted by $+$ and events for edge expiration are denoted by $-$. After building the query DAG, our algorithm performs steps 2 and 3 for each event in $L$ in chronological order.
    
    \setlength{\textfloatsep}{0pt}
\begin{algorithm}[t]
\small{
\KwIn{a temporal data graph $G$, a temporal query graph $q$, and a time window $\delta$}
\KwOut{all time-constrained embeddings}
$S_{max}\gets -\infty$\;
\ForEach{$r\in V(q)$}{
    $(\hat{q}_r, S_r)\gets$ \texttt{BuildDAG}$(q,r)$\;
    \If{$S_{max}<S_r$}{
        $S_{max}\gets S_r$\;
        $\hat{q}\gets \hat{q}_r$\;
    }
}
$g\gets$ an empty temporal graph\;
$L\gets \{(\overline{e},t,+),(\overline{e},t+\delta,-)\mid \overline{e}\in E(G), t=T_G(\overline{e})\}$\;
\While{$L\neq \emptyset$}{
    $(\overline{e},t,\textit{op})\gets$ pop from $L$ where $t$ is minimum\;
    \If{$\textit{op}=+$}{
        \texttt{InsertEdge}$(g, \overline{e})$\;
        $E_{DCS}^{+} \gets$\texttt{TCMInsertion}$(g, \hat{q}, \overline{e})$\;
        \texttt{DCSInsertion}$(\texttt{DCS},E_{DCS}^{+})$\;
        \texttt{FindMatches}$(\texttt{DCS},\overline{e},\emptyset )$\;
    }
    \If{$\textit{op}=-$}{
        \texttt{DeleteEdge}$(g, \overline{e})$\;
        $E_{DCS}^{-} \gets$\texttt{TCMDeletion}$(g, \hat{q}, \overline{e})$\;
        \texttt{DCSDeletion}$(\texttt{DCS},E_{DCS}^{-})$\;
        \texttt{FindMatches}$(\texttt{DCS},\overline{e},\emptyset )$\;    
    }
}

\caption{Time-Constrained Continuous Subgraph Matching}
\label{alg:overview}
\vspace*{-1mm}
}
\end{algorithm}
    
    \noindent\textbf{Updating the data structures.} 
    For each arrived edge $\overline{e}$, we update the data structures \textit{\texttt{DCS}} and \textit{max-min timestamp}. First, we update the data graph $g$ that represents the current state of $G$ by inserting the edge into $g$ (Line 12). We store the edges in an adjacency list in the chronologically sorted order. When inserting an edge, therefore, simply adding it to the end of the adjacency list is sufficient. Next, we update the max-min timestamp by invoking procedure \texttt{TCMInsertion} (Line 13). This procedure returns a set $E^+_{DCS}$ consisting of pairs of edges $e'\in E(q)$ and $\overline{e'}\in E(g)$ such that $e'$ newly becomes a time-constrained matchable edge of $\overline{e'}$ due to the arrived edge. Finally, we update the auxiliary data structure \texttt{DCS} with $E^+_{DCS}$ (Line 14). \texttt{DCS} is the data structure used in \cite{min2021symmetric} to solve the continuous subgraph matching problem. It stores an intermediate result of filtering vertices using the weak embedding of a DAG, which is a necessary condition for embedding. When filtering vertices, \texttt{DCS} in \cite{min2021symmetric} considers all pairs of edges in $g$ and $q$, but our \texttt{DCS} considers only the edge pairs remaining after the filtering by time-constrained matchable edges. 
    Because fewer edge pairs are used for our \texttt{DCS}, the overall running time to update \texttt{DCS} and the number of remaining candidate vertices after filtering are reduced, compared to \texttt{DCS} in \cite{min2021symmetric}.
    In a similar process, the data graph $g$ and the data structures can be updated for expired edges (Lines 17--19). \texttt{DeleteEdge} can be done by removing the edge from the front of the adjacency list.
    The pseudocodes of \texttt{DCSInsertion} and \texttt{DCSDeletion} are available in [23], as these operations are related to the update of the data structure \texttt{DCS} proposed in [23].

    \noindent\textbf{Matching.} After updating the data structures, we find new or expired time-constrained embeddings from the updated data structures by calling the backtracking procedure \icdepaper{\texttt{FindMatches} (Section \ref{sec:backtracking}).}\fullpaper{\texttt{FindMatches.}}
    \fullpaper{\texttt{FindMatches} is based on the backtracking algorithm that solve the continuous subgraph matching problem. Since the continuous subgraph matching problem deals with non-temporal graphs, it doesn't matter which one among parallel edges a query edge matches. So the continuous subgraph matching algorithms generally assume simple graphs and focus only on mapping of vertices. In contrast, in our problem, the search tree varies depending on which one of the parallel edges matches a query edge due to the temporal order. Therefore, we consider parallel edges and the temporal order during backtracking (Section \ref{sec:backtracking}).}

\section{Filtering by Temporal Order}\label{sec:filtering_by_temporal_order_constraint}

In this section, we introduce our edge filtering technique using the temporal order called time-constrained matchable edge and describe the greedy algorithm that generates a query DAG for effective filtering. Afterwards, we propose a data structure called max-min timestamp with polynomial space and an efficient way to update the data structure. This data structure is used to determine if an edge is filtered by the time-constrained matchable edge.

\subsection{Time-Constrained Matchable Edge}\label{subsec:time-constrained_matchable_edge}

There are various ways to filter edges considering the temporal order.
\texttt{Timing}~\cite{li2019time} and \texttt{Hasse}~\cite{sun2017hasse} store all time-constrained partial embeddings of the temporal query graph to prune some arrived edges, but it requires an exponential storage space and time to maintain and update partial embeddings. To avoid this, the TO-tree in \texttt{TOM}~\cite{li2021subgraph} only stores matching candidates of each query edge and filters candidates by considering temporal order between incident query edges. However, temporal relationships between non-incident edges are not utilized for filtering.
Our proposed method aims to occupy less storage space by not storing time-constrained partial embeddings and consider as many ordered pairs in the temporal order as possible for filtering.

\begin{MyDefinition}\label{def:time-constrained_weak_embedding}
Given a temporal data graph $G$, a temporal query graph $q$, and a query DAG $\hat{q}$, a weak embedding $M'$ of $\hat{q}_e$ at $\overline{e}\in E(G)$ is \textit{a time-constrained weak embedding} (\textit{TC-weak embedding} for short) if $T_G(M'(e))=T_G(\overline{e})<T_G(M'(e'))$ for every $e'$ where $e\leadsto_{\hat{q}}e'$. We say that $e$ is a \textit{time-constrained matchable edge} (\textit{TC-matchable edge} for short) of $\overline{e}$ regarding $\hat{q}$ if there exists a TC-weak embedding of $\hat{q}_{e}$ at $\overline{e}$.
\end{MyDefinition}

\begin{MyExample}\label{ex:time-constrained_matchable_edge}
Let us consider $\varepsilon_2$ in $q$ (Figure \ref{fig:temporal_query_graph_q}) and $\sigma_8$ in $G$ (Figure \ref{fig:temporal_data_graph_g}). There exists a weak embedding $M'_1=\{ (\varepsilon_2,\sigma_8),(\varepsilon_4,\sigma_{13}),(\varepsilon_5,\allowbreak\sigma_{10}),(\varepsilon_6,\sigma_{14})\}$ of $\hat{q}_{\varepsilon_2}$ at $\sigma_8$. This weak embedding is a TC-weak embedding of $\hat{q}_{\varepsilon_2}$ at $\sigma_8$ because $T_G(\sigma_8)<T_G(M'_1(\varepsilon_4))=T_G(\sigma_{13})$, $T_G(\sigma_8)<T_G(M'_1(\varepsilon_5))=T_G(\sigma_{10})$, and $T_G(\sigma_8)<T_G(\allowbreak M'_1(\varepsilon_6))=T_G(\sigma_{14})$ hold. Therefore, $\varepsilon_2$ is a TC-matchable edge of $\sigma_8$ regarding $\hat{q}$. 
On the other hand, there is no TC-weak embedding of $\hat{q}_{\varepsilon_2}$ at $\sigma_{12}$ since there is a constraint $\varepsilon_2 \prec \varepsilon_5$, but $\sigma_9$ and $\sigma_{10}$, which can be matched to $\varepsilon_5$, have timestamps not greater than $T_G(\sigma_{12})=12$. Hence, $\varepsilon_2$ is not a TC-matchable edge of $\sigma_{12}$ regarding $\hat{q}$.

\end{MyExample}

\begin{MyLemma}\label{lemma:time_constrained_filtering}
If $e\in E(q)$ is not a TC-matchable edge of $\overline{e}\in E(G)$ regarding some DAG of $q$, then there is no time-constrained embedding $M$ such that $M(e)=\overline{e}$.
\end{MyLemma}
\fullpaper{
\begin{proof}
Suppose that there exists a time-constrained embedding $M$ such that $M(e)=\overline{e}$. Then, we need to show that $e$ is a TC-matchable edge of $\overline{e}$ for any DAG $\hat{q}$ of $q$. For any query DAG $\hat{q}$, consider the weak embedding $M'$ of $\hat{q}_{e}$ at $\overline{e}$ such that the image of each element under $M'$ is the same as the image under $M$. For a temporal descendant $e'$ of $e$ in $\hat{q}$, $T_G(M'(e))=T_G(\overline{e})<T_G(M'(e'))$ holds by the definition of time-constrained embedding. Thus, $M'$ is a TC-weak embedding of $\hat{q}_e$ at $\overline{e}$ and so $e$ is a TC-matchable edge of $\overline{e}$ regarding $\hat{q}$.
\end{proof}
}

According to Lemma \ref{lemma:time_constrained_filtering}, if $e\in E(q)$ is not a TC-matchable edge of $\overline{e}\in E(G)$ regarding some DAG $\hat{q}$ of $q$, we can filter $\overline{e}$ 
from the set of data edges that are possible to match $e$.
For example, since $\varepsilon_2$ is not a TC-matchable edge of $\sigma_{12}$, we do not need to map $\varepsilon_2$ to $\sigma_{12}$ while backtracking, so the search space is reduced. In comparison to \texttt{TOM}~\cite{li2021subgraph}, which only checks the temporal order constraints between incident edges, \texttt{TOM} does not filter $\sigma_{12}$ because $\varepsilon_2$ and $\varepsilon_5$ are not incident. To enhance the filtering, we select the query DAG $\hat{q}$ containing as many temporal ancestor-descendant relationships as possible. We will describe in Section \ref{subsec:building_query_DAG} how to build the query DAG $\hat{q}$ to achieve this goal. Additionally, to prune more data edges which cannot be matched to $e$, we use both $\hat{q}$ and $\hat{q}^{-1}$ when filtering the candidates by Lemma~\ref{lemma:time_constrained_filtering}.

\subsection{Building Query DAG}\label{subsec:building_query_DAG}
Because the temporal ancestor-descendant relationship is affected by the shape of the query DAG, the results of filtering through TC-matchable edges are also affected by the shape of the query DAG. In particular, the ordered pair consisting of two edges with no path between them in the query DAG is not considered for filtering. Therefore, we propose a greedy algorithm (Algorithm \ref{alg:buildDAG}) that builds a query DAG so that there are many ordered pairs in the temporal ancestor-descendant relationships.

Algorithm \ref{alg:buildDAG} that makes a DAG with $r\in V(q)$ as its root proceeds as follows. 
\deleted{Initially, we set the query DAG $\hat{q}_r$ to an empty graph and the score of the DAG $S_r$ to 0. In addition, we initialize the set $cand$ of query vertices to $\{r\}$ and the score of each vertex $Score[u]$ to 0.}
The vertices in $cand$ are the candidates to come next in the topological order of the current DAG. $Score[u]$ stores the number of ordered pairs that will have temporal ancestor-descendant relationships when the DAG is constructed by selecting $u$ as the next vertex from $cand$. Now, we repeat the process of selecting a query vertex $u$ from $cand$, adding $u$ with directed edges into the current DAG, and updating $cand$ and $Score$ until there are no more vertices in $cand$. Specifically, we greedily select $u$ with the highest $Score[u]$ among the candidates in $cand$ so that the score of the query DAG can be high. If there is more than one vertex with the highest score, we select the first vertex inserted into $cand$ among them. 
\deleted{Then, we add the selected vertex $u$ into $\hat{q}_r$. Next, for each edge $(u,u')\in E(q)$ incident on $u$, if $u'$ already exists in $\hat{q}_{r}$, we add a directed edge from $u'$ to $u$ because $u'$ comes before $u$ in the topological order of $\hat{q}_{r}$. If $u'$ doesn't exist in $\hat{q}_{r}$, we insert $u'$ into $cand$ if $u'$ is not in $cand$ and compute $Score[u']$.}
Finally, the score of the DAG $S_r$ is the sum of $Score[u]$. We invoke Algorithm \ref{alg:buildDAG} with every vertex of the temporal query graph as the root once and then select the DAG with the highest score. 
\deleted{The query DAG is also used in \texttt{DCS}, so if there are multiple query DAGs with the highest score, the one with the highest height is chosen, as in \cite{min2021symmetric}. We use the query DAG built through the above process and its reverse for filtering by the TC-matchable edge and \texttt{DCS}.}

\setlength{\textfloatsep}{0pt}
\begin{algorithm}[t]
\small{
\KwIn{a temporal query graph $q$ and a root vertex $r$}
\KwOut{a query DAG $\hat{q}_r$ and the score $S_r$ of $\hat{q}_r$}

$\hat{q}_r\gets$an empty graph, $S_r\gets 0$, $cand\gets \{r\}$\;
$Score[u]\gets 0$ for $u\in V(q)$\;
\While{$cand\neq \emptyset$}{
     $u\gets$ pop from $cand$ whose $Score[u]$ is maximum\;
     add $u$ into $V(\hat{q}_r)$\;
     \ForEach{$(u,u')\in E(q)$}{
        \If{$u'\in V(\hat{q}_r)$}{
            add an edge $\overrightarrow{(u',u)}$ into $E(\hat{q}_r)$\;
        }
        \Else{
            \If{$u'\not\in cand$}{
                $cand\gets cand\cup \{u'\}$
            }
            compute $Score[u']$\;
        }
     }
     $S_r\gets S_r+Score[u]$\;
}
\KwRet{$(\hat{q}_r,S_r)$}\;

\caption{BuildDAG}
\label{alg:buildDAG}
\vspace*{-1mm}
}
\end{algorithm}

\fullpaper{
\begin{MyExample}
Figure \ref{fig:buildDAG} shows some processes of building $\hat{q}$ in Figure \ref{fig:DAG_q} from $q$ in Figure \ref{fig:temporal_query_graph_q} through Algorithm \ref{alg:buildDAG} with $u_1$ as the root. After $u_1$ is selected from $cand$ and processed, there are two candidates $u_2$ and $u_3$ in $cand$. $Score[u_2]$ is 1 because $\varepsilon_1$ becomes a temporal ancestor of $\varepsilon_3$ if $u_2$ is selected as the next vertex. Similarly, $Score[u_3]$ is 2 because of two ordered pairs $\varepsilon_2 \prec \varepsilon_4$ and $\varepsilon_2 \prec \varepsilon_6$. Therefore, $u_3$ with a higher score is selected among $u_2$ and $u_3$, and the DAG is changed as shown in Figure \ref{fig:buildDAG1}. Then, given $Score[u_2]=Score[u_4]=1$ and $Score[u_5]=0$, we choose $u_2$ because $u_2$ was inserted into $cand$ before $u_4$, so the DAG in Figure \ref{fig:buildDAG2} is created. Next, we choose $u_4$ whose score is 2 and then choose $u_5$ whose score is 0. Finally, we return the DAG in Figure \ref{fig:DAG_q} with a score of 5 (=2+1+2).
\end{MyExample}

\begin{figure}
\centering
    \subcaptionbox{After selecting $u_3$\label{fig:buildDAG1}}{
        \includegraphics[width=0.29\linewidth]{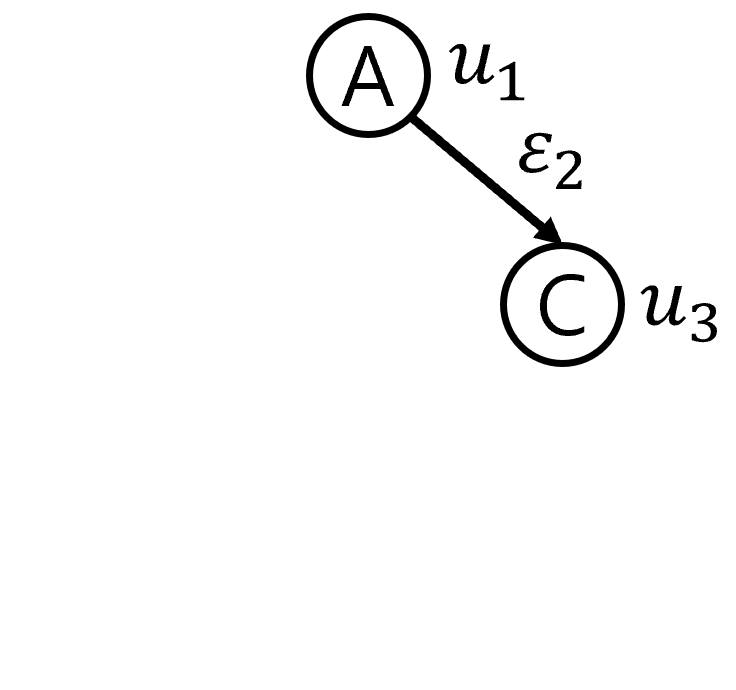}
    }
    \subcaptionbox{After selecting $u_2$\label{fig:buildDAG2}}{
        \includegraphics[width=0.29\linewidth]{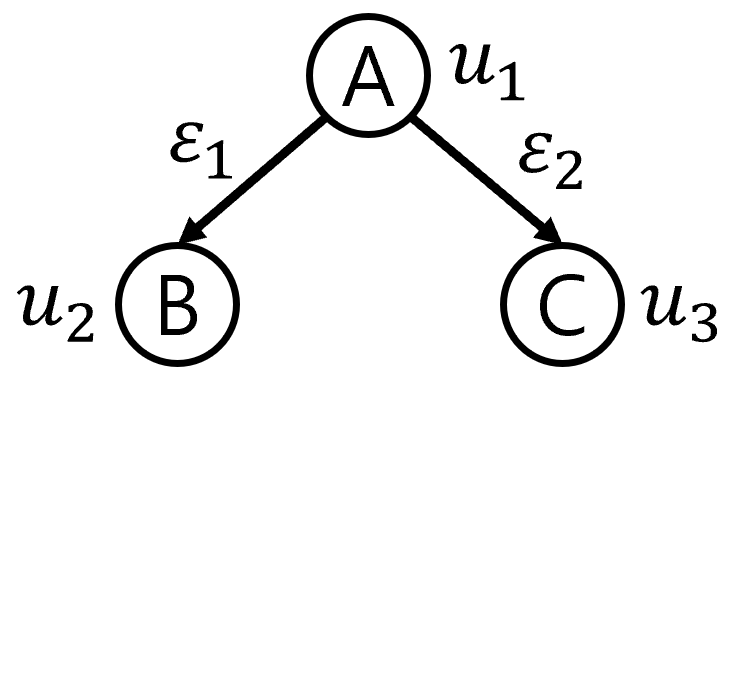}
    }
    \subcaptionbox{After selecting $u_4$\label{fig:buildDAG3}}{
        \includegraphics[width=0.29\linewidth]{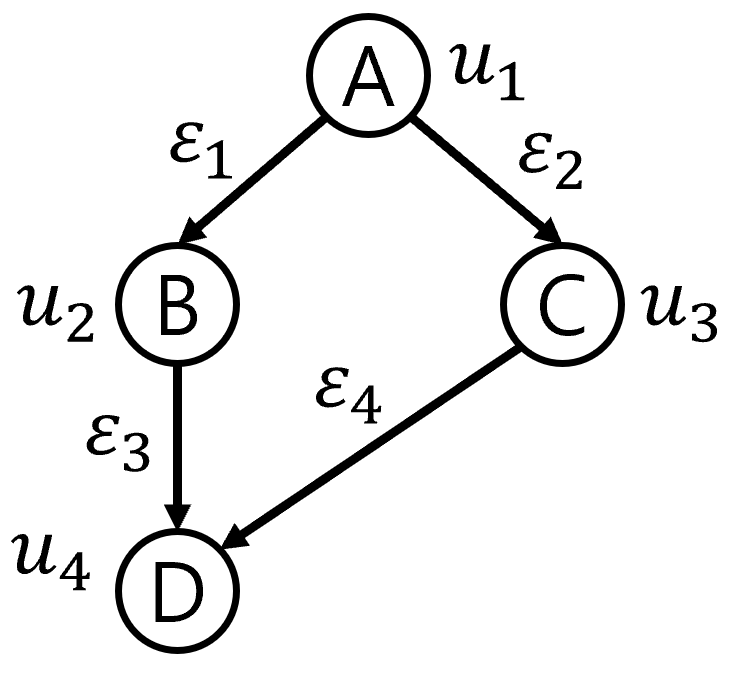}
    }
\caption{{Process of building DAG from $q$ in Figure \ref{fig:temporal_query_graph_q} with $u_1$ as the root}}
\label{fig:buildDAG}
\end{figure}
}

\begin{MyLemma}
Given a query graph $q$, the time complexity of Algorithm \ref{alg:buildDAG} is $O(\sum_{u \in V(q)} \deg(u)^2\times |E(q)|)$ where $\deg(u)$ represents the degree of vertex $u$. 
\end{MyLemma}
\fullpaper{
\begin{proof}
Lines 1 and 2 can be done in $O(1)$ time and $O(|V(q)|)$ time, respectively. The while loop is iterated $|V(q)|$ times, as each vertex in $V(q)$ can only be added to $cand$ once. Lines 5 and 13 can be executed in constant time. The process of pushing and popping vertex in $cand$ takes $O(\log{|V(q)|})$ time for each operation (Lines 4 and 11). Lines 7-8, which add an edge to the DAG $\hat{q}_r$, take a total of $O(|E(q)|)$ time as it is performed for each edge. Therefore, the time complexity of the algorithm up to this point is $O(|V(q)|\log{|V(q)|}+|E(q)|)$.

The remaining part involves the computation of $Score[u']$ (Line 12). To compute $Score[u']$, we iterate over the neighbors $u'_n$ of $u'$ and count the number of temporal ancestors for $(u',u'_n)$ in the DAG $\hat{q}_r$ if $u'_n$ does not belong to $V(\hat{q}_r)$. Since counting the number of temporal ancestors for each edge requires $O(|E(q)|)$ time and there are at most $\deg(u')$ neighbors, Line 12 takes $O(\deg(u')\times |E(q)|)$ to execute once. Line 12 is executed when visiting an edge $(u,u')$, and since each edge is visited at most once, Line 12 is performed up to $\deg(u')$ times for the vertex $u'$. Thus, for vertex $u'$, the total run time of Line 12 becomes $O(\deg(u')^2 \times |E(q)|)$. Since vertices are visited only once, Line 12 requires $O(\sum_{u \in V(q)} \deg(u)^2 \times |E(q)|)$ time in total. Consequently, the time complexity of Algorithm~\ref{alg:buildDAG} is $O(|V(q)|\log{|V(q)|}+|E(q)|+\sum_{u \in V(q)} \deg(u)^2 \times |E(q)|$), which simplifies to $O(\sum_{u \in V(q)} \deg(u)^2 \times |E(q)|)$. 
\end{proof}
}

Algorithm \ref{alg:buildDAG} is called $|V(q)|$ times in Algorithm \ref{alg:overview}, so the total time complexity of building DAG is $O(\sum_{u \in V(q)} \deg(u)^2\times |V(q)| \times |E(q)|)$.

\subsection{Computing TC-Matchable Edge}\label{subsec:computing_time-constrained_matchable_edge}

In this subsection, we first propose a data structure \textit{max-min time\-stamp}, which helps determine whether TC-weak embeddings exist or not and requires polynomial space. Also, we describe how to efficiently update the data structure when a new edge arrives or an existing edge expires in the data graph. There are two main points we considered when designing the data structure.

One is that when the data graph changes due to a newly arrived or expired edge $\overline{e}$, it does not only affect the TC-matchable edges for $\overline{e}$, but also affects other TC-matchable edges. For example, $\varepsilon_2$ of $q$ in Figure \ref{fig:temporal_query_graph_q} becomes a TC-matchable edge of $\sigma_8$ of $G$ in Figure \ref{fig:temporal_data_graph_g} regarding $\hat{q}$ after $\sigma_{14}$ arrives. However, since it is expensive to recompute the filtering results for all edges, it is necessary to focus on the part affected by the changed edge $\overline{e}$. 

The other point is that it is inefficient to search for TC-weak embeddings from scratch to determine whether a query edge is a TC-matchable edge of a data edge for each update. If we store information about partial TC-weak embeddings we found before, we can utilize it to determine if there is a TC-weak embedding. For instance, when determining whether $\varepsilon_2$ is a TC-matchable edge of $\sigma_8$ after the arrival of $\sigma_{14}$, we need to find the TC-weak embedding $M'=\{(\varepsilon_2,\sigma_8),(\varepsilon_4,\sigma_{13}),(\varepsilon_5,\allowbreak\sigma_{10}), (\varepsilon_6,\sigma_{14})\}$. A part of $M'$, $\{(\varepsilon_4,\sigma_{13}),(\varepsilon_5,\allowbreak\sigma_{10})\}$, is created when $\sigma_{13}$ arrives, so if we store the information about that part, we can use it to find $M'$ when $\sigma_{14}$ arrives.

\begin{MyDefinition}\label{def:representative_timestamp}
Given a weak embedding $M'$ of $\hat{q}_u$, the \textit{min time\-stamp for $e\in E(q)$} of $M'$  is defined as the minimum timestamp among $M'(e')$ where $e'$ is a temporal descendant of $e$ in $\hat{q}$. If there is no such $e'$, it is defined as $\infty$. 
\end{MyDefinition}

\begin{MyDefinition}\label{def:array_E}
For each $u\in V(q)$, $v\in V(G)$, and $e\in E(q)$, the \textit{max-min timestamp for $e$ of $\hat{q}_u$ at $v$}, denoted by $\mathcal{T}(\hat{q})[u,v,e]$, is the largest min timestamp for $e$ among weak embeddings of $\hat{q}_u$ at $v$. If there is no weak embedding of $\hat{q}_u$ at $v$, it is defined as $-\infty$. For each $e'\in E(q), \overline{e'}\in E(G)$, and $e\in E(q)$, similarly, $\mathcal{T}(\hat{q})[e',\overline{e'},e]$ denotes the \textit{max-min timestamp for $e$ of $\hat{q}_{e'}$ at $\overline{e'}$}, which is the largest min timestamp for $e$ among weak embeddings of $\hat{q}_{e'}$ at $\overline{e'}$. For brevity, we will use the simplified notation $\mathcal{T}$ instead of $\mathcal{T}(\hat{q})$ when the context is unambiguous.
\end{MyDefinition}

\begin{MyExample}\label{ex:array_E}
In a temporal data graph $G$ in Figure \ref{fig:temporal_data_graph_g}, a temporal query graph $q$ in Figure \ref{fig:temporal_query_graph_q}, and a query DAG $\hat{q}$ in Figure \ref{fig:DAG_q}, there exists four weak embeddings $\{(\varepsilon_4,\sigma_{13}),\allowbreak(\varepsilon_5,\sigma_9), (\varepsilon_6,\sigma_{7})\}$, $\{(\varepsilon_4,\sigma_{13}),\allowbreak(\varepsilon_5,\sigma_9), (\varepsilon_6,\sigma_{14})\}$, $\{(\varepsilon_4,\sigma_{13}),\allowbreak(\varepsilon_5,\sigma_{10}), (\varepsilon_6,\sigma_{7})\}$, and $\{(\varepsilon_4,\sigma_{13}),(\varepsilon_5,\allowbreak\sigma_{10}), (\varepsilon_6,\sigma_{14})\}$ of $\hat{q}_{u_3}$ at $v_4$. Since $\varepsilon_4$, $\varepsilon_5$, and $\varepsilon_6$ are temporal descendants of $\varepsilon_2$, the min timestamps for $\varepsilon_2$ of each weak embedding are 7, 9, 7, and 10, respectively. Thus, $\mathcal{T}[u_3,\allowbreak v_4,\varepsilon_2]$ stores the maximum value of 10 among them. 
\end{MyExample}

\setlinepenalty{80}

For a TC-weak embedding $M'$ of $\hat{q}_{e}$ at $\overline{e}$, the min timestamp for $e$ of any partial weak embedding of $M'$ is greater than the timestamp of $\overline{e}$ by definition of TC-weak embedding. That is, if there is no partial weak embedding with a min timestamp for $e$ greater than the timestamp of $\overline{e}$, there is no TC-weak embedding of $\hat{q}_{e}$ at $\overline{e}$. 
Therefore, we can use the max-min timestamps $\mathcal{T}$ to determine if $e\in E(q)$ is a TC-matchable edge of $\overline{e}\in E(G)$ regarding $\hat{q}$. 
In Example \ref{ex:array_E}, since $\mathcal{T}[u_3,v_4,\varepsilon_2]=10$, we can see that there is a weak embedding $M'_1$ of $\hat{q}_{u_3}$ at $v_4$ whose min timestamp for $\varepsilon_2$ is 10. To determine whether $\varepsilon_2=(u_1,u_3)$ is a TC-matchable edge of $\sigma_8=(v_1,v_4,8)$ regarding $\hat{q}$, consider a weak embedding $M'_2$ of $\hat{q}_{\varepsilon_2}$ at $\sigma_8$ by adding $\{(\varepsilon_2,\sigma_8)\}$ to the previous weak embedding $M'_1$. Then, $T_G(M'_2(\varepsilon_2))=T_G(\sigma_8)=8<10\leq T_G(M'_2(\varepsilon))$ holds for every $\varepsilon$ where $\varepsilon_2\leadsto_{\hat{q}} \varepsilon$ because $T_G(M'_2(\varepsilon))$ is greater than or equal to 10 which is the min timestamp for $\varepsilon_2$ of $M'_1$. Thus, $M'_2$ is a TC-weak embedding $\hat{q}_{\varepsilon_2}$ at $\sigma_8$ that we want to find and $\varepsilon_2$ is a TC-matchable edge of $\sigma_8$ regarding $\hat{q}$.
On the other hand, $\varepsilon_2=(u_1,u_3)$ is not a TC-matchable edge of $\sigma_{12}=(v_1,v_4,12)$ regarding $\hat{q}$. This is because we can confirm that the min timestamp for $\varepsilon_2$ in all weak embeddings of $\hat{q}_{u_3}$ at $v_4$ is not greater than $T_G(\sigma_{12})=12$, which is derived from $\mathcal{T}[u_3,v_4,\varepsilon_2]=10$.
We can generalize this as follows.

\setlinepenalty{10}

\begin{MyLemma}\label{lemma:array_E_time-constrained_matchable_edge}
Given a temporal data graph $G$, a temporal query graph $q$, and a query DAG $\hat{q}$, $e=(u_1,u_2)\in E(q)$ is a TC-matchable edge of $\overline{e}=(v_1,v_2,t)\in E(G)$ regarding $\hat{q}$ if and only if the time\-stamp of $\overline{e}$ ($=t$) is less than $\mathcal{T}[u_2,v_2,e]$.
\end{MyLemma}
\fullpaper{
\begin{proof}
First, we show the `only if' part. Let $e$ be a TC-matchable edge of $\overline{e}$, implying the existence of a TC-weak embedding $M$ of $\hat{q}_e$ at $\overline{e}$. According to Definition \ref{def:time-constrained_weak_embedding}, for every $e'$ such that $e\leadsto_{\hat{q}}e'$, it holds that $T_G(M(e'))$ is greater than $T_G(M(e))=T_G(\overline{e})=t$ . That is, for a weak embedding $M'=M-\{(e, \overline{e})\}$ of $\hat{q}_{u_2}$ at $v_2$, the min timestamp for $e$ of $M'$ ($=t'$) is greater than $t$. Since $\mathcal{T}[u_2, v_2, e]$ is the largest min timestamp for $e$ among weak embeddings of $\hat{q}_{u_2}$ at $v_2$, it follows that $\mathcal{T}[u_2, v_2, e] \ge t' > t$.\\
Now, we demonstrate the `if' part. According to Definition \ref{def:array_E}, there exists a weak embedding $M'$ of $\hat{q}_{u_2}$ at $v_2$ with a min timestamp for $e$ equal to $\mathcal{T}[u_2, v_2, e]$. By utilizing Definition \ref{def:representative_timestamp} and the given assumption, $T_G(M'(e')) \ge \mathcal{T}[u_2, v_2, e] > t$ holds for every $e'$ where $e\leadsto_{\hat{q}}e'$. Hence, the weak embedding $M=M' \cup \{(e, \overline{e})\}$ of $\hat{q}_e$ at $\overline{e}$ satisfies the conditions of a TC-weak embedding, as defined in Definition \ref{def:time-constrained_weak_embedding}. Therefore, $e$ is a TC-matchable edge of $\overline{e}$ and we proved the statement.
\end{proof}
}

According to Lemma \ref{lemma:array_E_time-constrained_matchable_edge}, we can determine at constant time whether $e\in E(q)$ is a TC-matchable edge of $\overline{e}\in E(G)$ regarding $\hat{q}$, given an array $\mathcal{T}$. Therefore, when a change occurs in the data graph, it is sufficient to update $\mathcal{T}$ without searching TC-weak embeddings. In order to solve the two issues mentioned at the beginning of the subsection, we should recompute only the portion of $\mathcal{T}$ whose values may change, and also utilize the portion of $\mathcal{T}$ already computed when recomputing $\mathcal{T}[u,v,e]$.

We address the second issue by using the fact that a weak embedding of $\hat{q}_u$ consists of weak embeddings of $\hat{q}_{(u,u_c)}$ where $u_c$ is a children of $u$.
For each $(u,u_c)$, consider a weak embedding with the largest min timestamp for $e$ among weak embeddings of $\hat{q}_{(u,u_c)}$. Then the weak embedding of $\hat{q}_u$ composed of such weak embeddings also has the largest min timestamp for $e$.
Furthermore, a weak embedding of $\hat{q}_{(u,u_c)}$ can be divided into the mapping of $(u,u_c)$ and the weak embedding of $\hat{q}_{u_c}$. Since $\mathcal{T}[u_c,\cdot,e]$ stores the max-min timestamps for $e$ among weak embeddings of $\hat{q}_{u_c}$, it can be utilized to compute $\mathcal{T}[u,v,e]$. Given $\mathcal{T}[u_c,\cdot,e]$ for every $u_c$, we can compute $\mathcal{T}[u,v,e]$ as follows. 

First, assuming that $(u,u_c)$ matches $(v,v_c,t)$, $\mathcal{T}[(u,u_c),(v,v_c,\allowbreak t),e]$, which is the largest min time\-stamp for $e$ among weak embeddings of $\hat{q}_{(u,u_c)}$ at $(v,v_c,t)$, can be obtained from $\mathcal{T}[u_c,v_c,e]$ that we have. If $(u,u_c)$ is not a temporal descendant of $e$, then $\mathcal{T}[(u,u_c),(v,v_c,t),e]$ is not related to $(u,u_c)$. So, $\mathcal{T}[(u,u_c),\allowbreak (v,v_c,t),e]$ is the same as the max-min timestamp for $e$ of $\hat{q}_{u_c}$ at $v_c$ (i.e., $\mathcal{T}[u_c,v_c,e]$). If $(u,u_c)$ is a temporal descendant of $e$, $\mathcal{T}[(u,u_c),(v,v_c,t),\allowbreak e]$ is the smaller of $\mathcal{T}[u_c,v_c,e]$ and $t$. The following equation summarizes this.
\begin{equation}\nonumber
    \mathcal{T}[(u,u_c),(v,v_c,t),e] =
    \begin{cases}
        \min(\mathcal{T}[u_c,v_c,e],t) & \hspace{-2mm}\text{if $e\leadsto_{\hat{q}} (u,u_c)$}\\
        \mathcal{T}[u_c,v_c,e] & \hspace{-2mm}\text{otherwise}
    \end{cases}
\end{equation}
Now, consider a weak embedding of $\hat{q}_{(u,u_c)}$ for $\mathcal{T}[u,v,e]$. Since there may be several edges incident on $v$ such that $(u,u_c)$ can be matched, the largest min timestamp for $e$ among weak embeddings of $\hat{q}_{(u,u_c)}$ is the maximum of the max-min timestamps when $(u,u_c)$ is matched to each edge. Finally, if the largest min timestamp is computed for all children of $u$, $\mathcal{T}[u,v,e]$ becomes the smallest of them. Summarizing this process, the following recurrence can be obtained.

\begin{equation}
    \mathcal{T}[u,v,e]=\min_{u_c\in \textrm{Child}(u)}(\max_{(v,v_c,t)}(\mathcal{T}[(u,u_c),(v,v_c,t),e]))\label{eqn:recurrence_E}
\end{equation}
Based on the above recurrence, we can compute $\mathcal{T}$ by dynamic programming in a bottom-up fashion in DAG $\hat{q}$.

What is now left is to recompute only the portions that may change, not the whole of $\mathcal{T}$, when $\mathcal{T}$ needs to be updated due to the arrival or expiration of an edge. 
We handle this by recomputing $\mathcal{T}[u,v,e]$ for only two cases: {(\romannumeral 1)} when an edge $(v,v_c,t)$ that matches $(u,u_c)$ arrives or expires, or {(\romannumeral 2)} when $\mathcal{T}[u_c,v_c,e]$ changes in the process of updating $\mathcal{T}$.

\setlength{\textfloatsep}{0pt}
\begin{algorithm}[t]
\small{
\KwIn{a data graph $g$, a query DAG $\hat{q}$, and a new edge $\overline{e}=(v_1,v_2,t)$}
\KwOut{a set of edge pairs $E^{+}_{DCS}$}

$E^{+}_{DCS}\gets\emptyset$\;
$Q\gets$ empty queue\;
\ForEach{$e=(u_1,u_2)\in E(\hat{q})$ \textup{that matches to $\overline{e}$}}{
    \If{$t<\mathcal{T}[u_2,v_2,e]$}{
        $E^{+}_{DCS}\gets E^{+}_{DCS}\cup \{(e,\overline{e})\}$
    }
    \ForEach{\textup{ancestor $e'$ of $e$}}{
        Compute $\mathcal{T}[u_1,v_1,e']$ by Eq. (1)\;
        \If{$\mathcal{T}[u_1,v_1,e']$ \textup{changes}}{
            $Q.push((u_1,v_1,e'))$
        }
    }
    \While{$Q\neq \emptyset$}{
        $(u,v,e')\gets Q.pop$\;
        \ForEach{$u_p\in \textup{Parent}(u)$}{
            \ForEach{$(v_p,v,t')$ \textup{that $(u_p,u)$ can match}}{
                \If{$e'=(u_p,u)$ \textup{and} $t'<\mathcal{T}[u,v,e']$ \textup{and} $((u_p,u),(v_p,v,t'))\notin DCS$}{
                    $E^{+}_{DCS}\gets E^{+}_{DCS}\cup \{((u_p,u),(v_p,v,t'))\}$
                }
                \If{$e'$ \textup{is an ancestor of $(u_p,u)$}}{
                    Compute $\mathcal{T}[u_p,v_p,e']$ by Eq. (1)\;
                     \If{$\mathcal{T}[u_p,v_p,e']$ \textup{changes}}{
                        $Q.push((u_p,v_p,e'))$
                    }
                }
                
            }
        }
    }
}
\KwRet{$E^{+}_{DCS}$}\;

\caption{TCMInsertion}
\label{alg:TCMinsertion}
\vspace*{-1mm}
}
\end{algorithm}

Algorithm~\ref{alg:TCMinsertion} shows the process of updating the portion of $\mathcal{T}$ through the above method and Equation~(\ref{eqn:recurrence_E}) when a new edge $\overline{e}$ arrives. It performs the update process of cases {(\romannumeral 1)} (Lines 4--9) and {(\romannumeral 2)} (Lines 10--19). This procedure also returns a set $E^+_{DCS}$ consisting of pairs of edges $e'\in E(q)$ and $\overline{e'}\in E(g)$, where $e'$ newly becomes a time-constrained matchable edge of $\overline{e'}$ in the process of updating $\mathcal{T}$. 
\deleted{First, it initializes $E^+_{DCS}$ as an empty set and $Q$ as an empty queue (Lines 1--2). 
For an arrived edge $\overline{e}=(v_1, v_2, t)$ and each edge $e=(u_1, u_2) \in E(\hat{q})$ such that $e$ matches $\overline{e}$ (i.e., the endpoints of $e$ and $\overline{e}$ have the same label), $(e, \overline{e})$ is inserted into $E^+_{DCS}$ if $e$ is a TC-matchable edge of $\overline{e}$ by Lemma \ref{lemma:array_E_time-constrained_matchable_edge} (Lines 3--5). Next, for case {(\romannumeral 1)}, it recomputes $\mathcal{T}[u_1, v_1, e']$ for every ancestor $e'$ of $e$ in $\hat{q}$ (Lines 6--7). To handle case {(\romannumeral 2), if $\mathcal{T}[u_1, v_1, e']$ changed, it pushes $(u_1, v_1, e')$ into $Q$ (Lines 8--9).}  Finally, for case {(\romannumeral 2)}, it recomputes the remaining portion of $\mathcal{T}$ by processing all tuples in $Q$. For a tuple $(u, v, e')\in Q$ and an edge pair $((u_p, u), (v_p, v, t'))$ such that $(u_p, u)$ matches $(v_p, v, t')$, it performs a similar process to update $E^+_{DCS}$ and $\mathcal{T}[u_p, v_p, e']$ until $Q$ becomes empty (Lines 10--19).} \texttt{TCMDeletion}, which updates $\mathcal{T}$ when an edge expires, is the same as \texttt{TCMInsertion} except that Line 14 checks if $((u_p,u),(v_p,v,t'))$ is in \texttt{DCS}.

\begin{MyExample}\label{ex:compute_E}
Consider when $\sigma_{14}=(v_4,v_7,14)$ in Figure \ref{fig:temporal_data_graph_g} arrives. Since $\varepsilon_6=(u_3,u_5)$ in Figure \ref{fig:temporal_query_graph_q} can match $\sigma_{14}$, we first check whether $\varepsilon_6$ is a TC-matchable edge of $\sigma_{14}$ regarding $\hat{q}$. Indeed, $\varepsilon_6$ is a TC-matchable edge of $\sigma_{14}$ regarding $\hat{q}$ because $T_G(\sigma_{14})<\mathcal{T}[u_5,v_7,\varepsilon_6]=\infty$. So, we insert $(\varepsilon_6,\sigma_{14})$ into $E_{DCS}^+$. Next, we recompute $\mathcal{T}[u_3,v_4,\varepsilon_2]$. For $u_4$, which is a child of $u_3$, there is only $\sigma_{13}=(v_4,v_5,13)$ which is incident on $v_4$ and can be matched to $\varepsilon_4=(u_3,u_4)$. Thus, the largest min timestamp for $\varepsilon_2$ among weak embeddings of $\hat{q}_{\varepsilon_4}$ is equal to $\mathcal{T}[\varepsilon_4,\sigma_{13},\varepsilon_2]=\min(\mathcal{T}[u_4,v_5,\varepsilon_2], 13)=\min(10,\allowbreak 13)=10$. For the other child $u_5$ of $u_3$, there are two edges $\sigma_7=(v_4,v_7,7)$ and $\sigma_{14}=(v_4,v_7,14)$ which are incident on $v_4$ and can be matched to $\varepsilon_6=(u_3,u_5)$. The max-min timestamps for $\varepsilon_2$ of $\hat{q}_{\varepsilon_6}$ when $\varepsilon_6$ is matched to $\sigma_7$ and $\sigma_{14}$ are $\mathcal{T}[\varepsilon_6,\sigma_{7},\varepsilon_2]=\min(\infty,7)=7$ and $\mathcal{T}[\varepsilon_6,\sigma_{14},\varepsilon_2]=\min(\infty,14)=14$, respectively. Therefore, the largest min timestamp for $\varepsilon_2$ among weak embeddings of $\hat{q}_{\varepsilon_6}$ is $\max(7,14)=14$. Finally, $\mathcal{T}[u_3,v_4,\varepsilon_2]$ becomes the minimum value of 10 among the timestamps 10 and 14 obtained from the children of $u_3$. As $\mathcal{T}[u_3,v_4,\varepsilon_2]$ is updated from 7 to 10, $\varepsilon_2$ becomes a TC-matchable edge of $\sigma_8$, so we insert $(\varepsilon_2,\sigma_8)$ into $E_{DCS}^+$. However, since the timestamp of $\sigma_{12}$ is still greater than $\mathcal{T}[u_3,v_4,\varepsilon_2]=10$, $\varepsilon_2$ does not become a TC-matchable edge of $\sigma_{12}$ and we do not have to insert $(\varepsilon_2,\sigma_{12})$ into $E_{DCS}^+$.
\end{MyExample}

\fullpaper{
\begin{MyLemma}\label{lemma:maxmin-timestamp-space-complexity}
    Given a temporal data graph $G$, a temporal query graph $q$, and a time window $\delta$, let $n_{max}$ and $m_{max}$ be the maximum number of vertices and edges in the temporal data graph $G$ within the time window $\delta$, respectively. Then the space complexity of $\mathcal{T}$ is $O(|E(q)|^2\times m_{max})$.
\end{MyLemma}
\begin{proof}
    There are two types of array: $\mathcal{T}[u,v,e]$ and $\mathcal{T}[e',\overline{e'},e]$. Regarding $\mathcal{T}[u,v,e]$, the variables $u$ and $e$ are bounded by $O(|V(q)|)$ and $O(|E(q)|)$, respectively. Furthermore, it is sufficient to store the max-min time\-stamps at $v$ within the current time window, rather than storing timestamps for all vertices of the temporal data graph. Therefore, the space complexity of $\mathcal{T}[u,v,e]$ is determined as $O(|V(q)|\times |E(q)|\times n_{max})$. Similarly, the space complexity of $\mathcal{T}[e',\overline{e'},e]$ can be expressed as $O(|E(q)|^2\times m_{max})$. Consequently, the overall space complexity of $\mathcal{T}$ is given by $O(|V(q)|\times |E(q)|\times n_{max} +|E(q)|^2\times m_{max})=O(|E(q)|^2\times m_{max})$. 
\end{proof}
}

\fullpaper{
\begin{MyLemma}
    Let $P$ be the set of $(u, v)$ where $\mathcal{T}[u, v, \cdot]$ changes. Then the time complexity of the \texttt{TCMInsertion} and \texttt{TCMDeletion} is $O(d_{max} \times \sum_{p \in P} \deg(p) + d_{max} \times |E(q)|^2)$, where $\deg(p)$ is the number of edges in \texttt{DCS} connected to $p$ \cite{min2021symmetric} and $d_{max}$ is the maximum value among $\deg(p)$.
\end{MyLemma}
\begin{proof}
    The function \texttt{TCMDeletion} is similar to the \texttt{TCMInsertion}, so we will only show the time complexity of \texttt{TCMInsertion} of Algorithm~\ref{alg:TCMinsertion}. First, Lines 4--5 are executed $O(|E(q)|)$ times and Lines 6--9 are executed $O(|E(q)|^2)$ times. Lines 4--5 and 8--9 take a constant time. Line 7 takes $O(d_{max})$ time to update $\mathcal{T}[u_1, v_1, e']$ since we need to check $\mathcal{T}[(u_1, u'), (v_1, v', \cdot), e']$ for every \texttt{DCS} edge $((u_1, u'), (v_1, v', \cdot))$ connected to $(u_1, v_1)$. Therefore, the total execution time of Lines 3--9 is $O(d_{max} \times |E(q)|^2)$. Next, the while loop of Lines 11--19 except Line 17 takes a time proportional to the number of parents of $(u, v)$ in \texttt{DCS}, which is equal to or less than $\deg((u, v))$. Line 17 is executed $O(\deg((u, v)))$ times, and it takes $O(d_{max})$ time for each execution. Since the while loop (Line 10) is executed for $(u, v)$ where $\mathcal{T}[u, v, \cdot]$ changes, the total execution time of Lines 11-19 is $O(d_{max} \times \sum_{p \in P} \deg(p))$. Hence, the time complexity of \texttt{TCMInsertion} is $O(d_{max} \times \sum_{p \in P} \deg(p) + d_{max} \times |E(q)|^2)$, and \texttt{TCMDeletion} has the same time complexity.
\end{proof}
}

\section{Time-Constrained Pruning in Backtracking}\label{sec:backtracking}
In this section, we present three time-constrained pruning techniques in backtracking.
Our matching algorithm gradually extends a partial embedding until it finds a time-constrained embedding like the existing backtracking-based algorithms that solve the continuous subgraph matching problem \cite{min2021symmetric, kim2018turboflux}. To extend a partial embedding $M$, existing backtracking-based algorithms select an unmapped vertex $u$ of a query graph $q$ and match $u$ to each candidate of $u$. 
In time-constrained matching, mapping between edges (rather than mapping between vertices) is essential because each edge has temporal information, and thus we use mapping between edges.
We find a set of candidate edges $(v,v')$ for an edge $(u,u')$ and then match $(u,u')$ to each candidate edge $(v,v')$ to extend a partial embedding $M$.
In backtracking, we prune some candidates $(v,v')$ of $(u,u')$ if it is guaranteed that there is no time-constrained embedding when $(u,u')$ matches $(v,v')$ by considering parallel edges and the temporal order. \fullpaper{Algorithm \ref{alg:FindMatches} shows this backtracking process.}

\fullpaper{
\setlength{\textfloatsep}{0pt}
\begin{algorithm}[t]
\small{
\KwIn{\texttt{DCS}, a data edge $\overline{e}$, and a partial time-constrained embedding $M$}
\KwOut{all occurred/expired time-constrained embeddings including $\overline{e}$}
\uIf{$|M|=|V(q)| + |E(q)|$}{
    Report $M$ as a match\;
}
\uElseIf{$|M|=0$}{
    Let $\overline{e}=(v, v', t)$\;
    \ForEach{$e=(u, u') \in E(q)$ such that $((u, u'), (v, v', t)) \in $ DCS}{
            $M \gets \{ (u, v), (u', v'), (e, \overline{e}) \}$\;
            \texttt{FindMatches}(\texttt{DCS}, $\overline{e}$, $M$)\;
    }
    
}
\Else{
    \eIf{$\exists$ an unmapped query edge with both endpoints mapped}{
        $e \gets$ an edge satisfying the above condition\;
        Compute a set of candidates $EC_M(e)$ of $e$\;
        \ForEach{$\overline{e} \in EC_M(e)$ which is not pruned}{
            $M' \gets M \cup \{(e, \overline{e})\}$\;
            \texttt{FindMatches}(\texttt{DCS}, $\overline{e}$, $M'$)\;
        }
    }{
        $u \gets$ next vertex according to the matching order in \cite{min2021symmetric}\;
        Compute a set of candidates $C_M(u)$ of $u$ as in \cite{min2021symmetric}\;
        \ForEach{$v \in C_M(u)$}{
            $M' \gets M \cup \{(u, v)\}$\;
            \texttt{FindMatches}(\texttt{DCS}, $\overline{e}$, $M'$)\;
        }
    }
}

\caption{FindMatches}
\label{alg:FindMatches}
\vspace*{-1mm}
}
\end{algorithm}
}

\begin{figure}
\centering
    \includegraphics[width=0.9\linewidth]{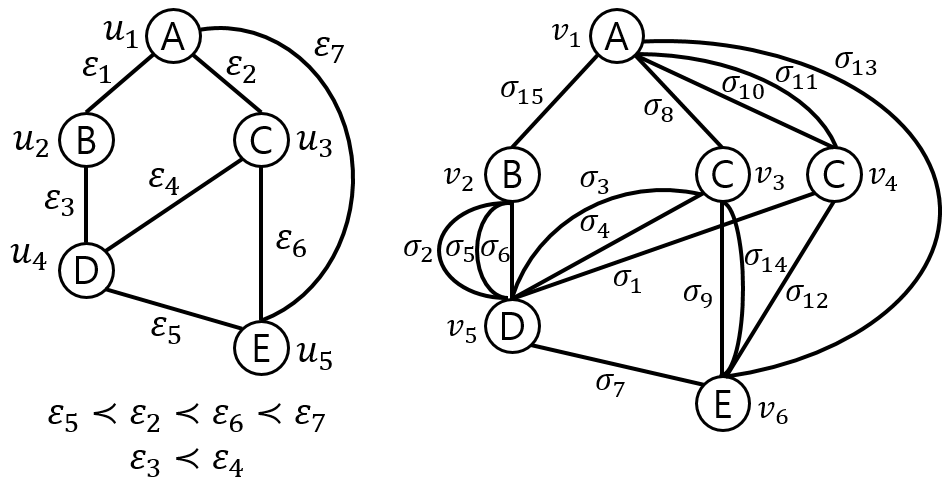}
\caption{{A new running example of a temporal query graph $q$ and a temporal data graph $G$}}
\label{fig:new_example}
\end{figure}

\begin{MyDefinition}\label{def:related_edges}
For each edge $e\in E(q)$ and a (partial) embedding $M$, a set $R_M^{+}(e)$ is the set of edges $e'\in E(q)$ that are temporally related to $e$ (i.e., $e\prec e'$ or $e'\prec e$) and are members of $M$. Conversely, a set $R_M^{-}(e)$ is the set of edges $e'\in E(q)$ that are temporally related to $e$ and are not members of $M$.
\end{MyDefinition}

\begin{MyDefinition}\label{def:extendable_edge_candidates}
A set $EC_M(e)$ of candidate edges of $e=(u_1,u_2)\in E(q)$ regarding a partial time-constrained embedding $M$ is defined as the set of edges in $E(G)$ between $M(u_1)$ and $M(u_2)$ that are TC-matchable edges for $e$ and satisfy the temporal relationships with $M(e')$ where $e'\in R_M^{+}(e)$.
\end{MyDefinition}

As a new running example, we use a temporal query graph $q$ and a temporal data graph $G$ in Figure \ref{fig:new_example}.

\begin{MyExample}
Consider the partial time-constrained embedding $M_1=\{(\varepsilon_1,\sigma_{15}),(\varepsilon_2,\sigma_8)\}$ and $\varepsilon_3$ as the next matching edge. Then $R_{M_1}^+(\varepsilon_3)=\emptyset$, $R_{M_1}^-(\varepsilon_3)=\{\varepsilon_4\}$, and $EC_{M_1}(\varepsilon_3)=\{\sigma_2,\sigma_5,\sigma_6\}$. If we extend $M_1$ to $M_2=M_1\cup \{(\varepsilon_3,\sigma_2)\}$ by matching $\varepsilon_3$ to $\sigma_2$, then $R_{M_2}^+(\varepsilon_4)=\{\varepsilon_3\}$, $R_{M_2}^-(\varepsilon_4)=\emptyset$, and $EC_{M_2}(\varepsilon_4)=\{\sigma_3,\sigma_4\}$ for the next matching edge $\varepsilon_4$. If $\varepsilon_3$ matches $\sigma_5$ instead of $\sigma_2$, then $EC_{M_2'}(\varepsilon_4)=\emptyset$ where $M_2'=M_1\cup \{(\varepsilon_3,\sigma_5)\}$ because $T_G(\sigma_5)\nless T_G(\sigma_3)$ and $T_G(\sigma_5)\nless T_G(\sigma_4)$ while $\varepsilon_3\prec \varepsilon_4$.
\end{MyExample}

When there are multiple candidate edges in $EC_M(e)$, we introduce a method to prune some candidate edges according to the three cases of $R_M^-(e)$ to reduce the search space. 

\fullpaper{
\begin{figure*}[t]
\centering
    \includegraphics[width=0.63\textwidth]{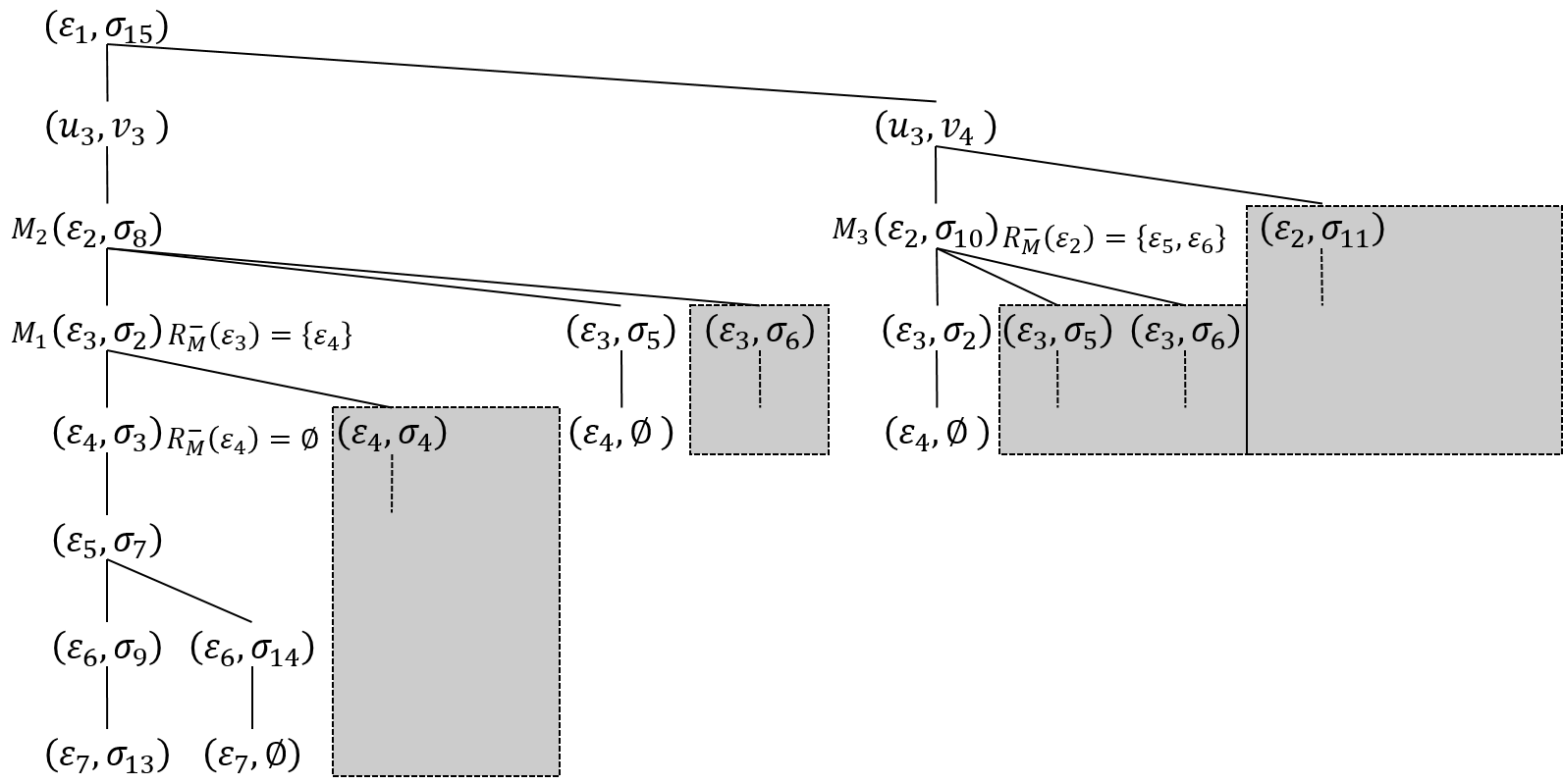}
\caption{{Search tree when $\sigma_{15}$ arrives. Nodes enclosed by dashed boxes are pruned}}
\label{fig:search_tree}
\end{figure*}
}

First, if $R_M^-(e)$ is an empty set (i.e., no temporally related edges remain), the search trees are the same no matter which edge of $EC_M(e)$ matches $e$ to extend $M$.
Therefore, we visit the node $M\cup \{(e,\overline{e})\}$ for only one candidate edge $\overline{e}\in EC_M(e)$ rather than extending $M$ over all candidate edges in $EC_M(e)$. When a time-constrained embedding is found in the subtree rooted at $M\cup \{(e,\overline{e})\}$, we can find other time-constrained embeddings without exploring other subtrees by matching $e$ to another candidate edge in $EC_M(e)$.

Second, consider $R_M^-(e)$ where all edges in $R_M^-(e)$ have the same temporal relationship with $e$. If $e\prec e'$ for all $e'\in R_M^-(e)$, we match $e$ to candidate edges in $EC_M(e)$ in chronological order. If no time-constrained embedding is found when we match $e$ to $\overline{e}$, we can skip the candidate edges after $\overline{e}$ in $EC_M(e)$. Conversely, when $e\succ e'$ for all $e'\in R_M^-(e)$, we match $e$ to candidate edges in $EC_M(e)$ in reverse chronological order.

Finally, when $R_M^-(e)$ is not in the previous two cases, we consider the case where we failed to find any time-constrained embedding in the search tree rooted at $M=\{\ldots, (e, \overline{e})\}$ whose last mapping is $(e, \overline{e})$, and prune the other candidate edges if possible. If $\overline{e}$ did not cause any failures related to the ordered pair $e$ and $e'\in R_M^-(e)$, we can prune other candidate edges of $e$. For each search tree node $M$, we compute the set of query edges related to the failures in the search tree rooted at $M$ and utilize it to prune other candidates.

\begin{MyDefinition}\label{def:temporal_failing_set}
Let $M$ be a search tree node whose last mapped query edge is $e$ and there is no time-constrained embeddings in the subtree rooted at $M$. A \textit{temporal failing set} $TF_M \subseteq E(q)$ of node $M$ is defined as follows:
\begin{enumerate}[1.]
    \item If the node $M$ is a leaf node (i.e., $M=\{\ldots, (e, \emptyset)\}$),  $TF_M=R_M^+(e)$.
    \item Otherwise,
    \begin{enumerate}[2.1]
        \item If there exists a child node $M_i=M \cup \{(e_i, \overline{e_i})\}$ of $M$ such that $e_i \notin TF_{M_i}$, we set $TF_M=TF_{M_i} \cup R_M^+(e)$.
        \item Otherwise, $TF_M=\bigcup_{i=1}^{k} TF_{M_i} \cup R_M^+(e)$, where $M_1,\ldots , M_k$ are the children of $M$.
    \end{enumerate}
\end{enumerate}
\end{MyDefinition}
For the search tree node $M$ whose last mapping is $(e, \overline{e})$, we can prune the other candidate edges of $e$ by testing whether $e$ is in the temporal failing set $TF_{M}$ of $M$. If $e \notin TF_{M}$, $\overline{e}$ did not cause any failures related to $e$, so the sibling nodes of $M$, i.e., $M - \{(e, \overline{e})\} \cup \{(e, \overline{e}')\}$ for every $\overline{e}' \in EC_M(e)$, can be pruned.\icdepaper{\footnote{An example for the pruning method is provided in the full version \cite{fullversion}.}}

\fullpaper{
\begin{MyExample}\label{ex:search_tree}
Figure \ref{fig:search_tree} shows the search tree of $q$ and $G$ in Figure \ref{fig:new_example} when $\sigma_{15}$ arrives. In this example, to focus on pruning techniques in backtracking, we do not consider the filtering in Section \ref{sec:filtering_by_temporal_order_constraint} and also use an arbitrary matching order. Except for $u_3$, since there is only one vertex in $G$ that can be matched to a query vertex, we omit the mappings of these vertices from the search tree. Furthermore, when representing a mapping $M$, we enumerate the pairs in $M$ in the order in which they are added to $M$.

First, consider the partial time-constrained embedding $M_1=\{(\varepsilon_1,\sigma_{15}),(\varepsilon_2,\sigma_8),(\varepsilon_3,\sigma_2)\}$ in Figure \ref{fig:search_tree} and $\varepsilon_4$ is to be matched next. At this time, since $R^-_{M_1}(\varepsilon_4)=\emptyset$ and $EC_{M_1}(\varepsilon_4)=\{\sigma_3,\sigma_4\}$, the subtrees rooted at $M_1\cup \{(\varepsilon_4,\sigma_3)\}$ and $M_1\cup \{(\varepsilon_4,\sigma_4)\}$ will be the same. Thus, we visit only the subtree rooted at $M_1\cup \{(\varepsilon_4,\sigma_3)\}$. When we find the time-constrained embedding $\{(\varepsilon_1,\sigma_{15}),(\varepsilon_2,\sigma_8),(\varepsilon_3,\sigma_2),(\varepsilon_4,\allowbreak\sigma_3),(\varepsilon_5,\allowbreak\sigma_7),(\varepsilon_6,\sigma_{9}),(\varepsilon_7,\sigma_{14})\}$ in that subtree, we replace $\sigma_3$ with $\sigma_4$ to find the other time-constrained embedding without visiting the other subtree rooted at $M_1\cup \{(\varepsilon_4,\sigma_4)\}$. 

Next, suppose that we came back to the node $M_2=\{(\varepsilon_1,\sigma_{15}),\allowbreak(\varepsilon_2,\sigma_8)\}$ after the exploration of the subtree rooted at $M_1=M_2\cup \{(\varepsilon_3,\sigma_2)\}$. Since all edges $\varepsilon$ in  $R^-_{M_2}(\varepsilon_3)=\{\varepsilon_4\}$ satisfy $\varepsilon_3\prec \varepsilon$, we match $\varepsilon_3$ to $\sigma_5$, which is the next edge of $EC_{M_2}(\varepsilon_3)=\{\sigma_2,\sigma_5,\sigma_6\}$ in chronological order. In the subtree rooted at $M_2\cup \{(\varepsilon_3,\sigma_5)\}$, there is no time-constrained embedding because there is no edge in $G$ where $\varepsilon_4$ can be matched (i.e., $EC_{M_2\cup \{(\varepsilon_3,\sigma_5)\}}(\varepsilon_4)=\emptyset$). When we return to the node $M_2$, we know that time-constrained embeddings do not exist even in the subtree rooted at $M_2\cup \{(\varepsilon_3,\sigma_6)\}$, so we prune that subtree.

Now, suppose that we explored the subtree rooted at $M_3=\{(\varepsilon_1,\sigma_{15}),(\varepsilon_2,\sigma_{10})\}$ and came back to the node $M_3$. We failed to find time-constrained embeddings and the failure occurred in $\varepsilon_4$, which is temporally related to $\varepsilon_3$. The same failure will occur even if we match $\varepsilon_2$ to $\sigma_{11}$ instead of $\sigma_{10}$ because the failure is related only to $\varepsilon_3$ and not to $\varepsilon_2$. Therefore, we prune the subtree rooted at $M_3-\{(\varepsilon_2,\sigma_{10})\}\cup\{(\varepsilon_2,\sigma_{11})\}$.
\end{MyExample}
}

\fullpaper{
\begin{MyLemma}
    The space complexity of Algorithm~\ref{alg:FindMatches} is $O(|E(q)| \times m_{max})$, where $n_{max}$ and $m_{max}$ are the same as defined in Lemma~\ref{lemma:maxmin-timestamp-space-complexity}.
\end{MyLemma}
\begin{proof}
    For each search tree node $M$, we map a vertex $u \in V(q)$ or an edge $e \in E(q)$ to its candidate vertex or candidate edge. If a vertex $u$ is mapped in $M$, we need $O(n_{max})$ space to store the set of candidate vertices $C_M(u)$ \cite{min2021symmetric} of $u$. Otherwise (i.e., an edge $e$ is mapped in $M$), we need $O(m_{max})$ space to store the set of candidate edges $EC_M(e)$ of $e$. Additionally, we store the temporal failing set $TF_M$ of $M$, which requires $O(|E(q)|)$ space. To sum up, each node requires $O(\max(n_{max}, m_{max}) + |E(q)|)$ space and the maximum depth of the search tree is $|V(q)| + |E(q)|$, resulting the space complexity $O((|E(q)| + m_{max}) \times |E(q)|)$ of Algorithm~\ref{alg:FindMatches}. 
    Since we need to invoke \texttt{FindMatches} only when $m_{max} \ge |E(q)|$, we obtain the space complexity $O(|E(q)| \times m_{max})$.
\end{proof}
}

\begingroup
\setlength{\tabcolsep}{4.5pt} %
\begin{table}[h]
    \centering
    \begin{tabular}{ccc}
        \toprule
        Function & Time complexity  & Space complexity\\
        \midrule
        \texttt{TCMUpdate} & $O(d_{max}\times\Sigma_{p\in P_1}\deg(p) +$  & $O(|E(q)|^2\times m_{max})$ \\
        &  $d_{max}\times|E(q)|^2)$ & \\
        \texttt{DCSUpdate} & $O(\Sigma_{p\in P_2}\deg(p) + |E_{DCS}^+|)$ & $O(|E(q)|\times n_{max})$ \\
        \texttt{FindMatches} & exponential & $O(|E(q)|\times m_{max})$\\
        \bottomrule
    \end{tabular}
    \caption{Time and space complexities of each function}
    \label{tab:complexity}
\end{table}
\endgroup

\icdepaper{
Table~\ref{tab:complexity} presents the time and space complexities of each function, where \texttt{Update} means both \texttt{Insertion} and \texttt{Deletion}. 
In time complexity, $P_1$ is the set of $(u, v)$ where the max-min timestamp $\mathcal{T}[u, v, \cdot]$ changes, $P_2$ is the set of $(u, v)$ where the value of $(u,v)$ in \texttt{DCS} changes, $\deg(p)$ is the number of edges in \texttt{DCS} connected to $p$, and $d_{max}$ is the maximum value among $\deg(p)$. In space complexity, $n_{max}$ and $m_{max}$ denote the maximum number of vertices and edges in the data graph $G$ within the time window $\delta$, respectively.}

The time and space complexities of \texttt{DCSUpdate} are the same as those in \cite{min2021symmetric}. 
Compared to \texttt{DCSUpdate}, \texttt{TCMUpdate} introduces an additional term $d_{max}$ to update $\mathcal{T}$ (Lines 7 and 17 of Algorithm~\ref{alg:TCMinsertion}). 
While \texttt{TCMUpdate} and \texttt{DCSUpdate} achieve polynomial time complexities, the exponential time complexity of \texttt{FindMatches} is inevitable due to the NP-hardness of the problem.  The space complexity of \texttt{TCMUpdate} 
is the size of the max-min timestamp $\mathcal{T}$, and
that of \texttt{FindMatches} is the depth of the search tree (i.e., $|E(q)|$) times the number of candidate edges $EC_M(e)$ for an edge $e \in E(q)$ which is bounded by $m_{max}$.

\fullpaper{
Table~\ref{tab:complexity} presents the time and space complexities of each function, where \texttt{Update} means both \texttt{Insertion} and \texttt{Deletion}. 
}

\section{Performance Evaluation}\label{sec:performance}

In this section, we evaluate the performance of our algorithm against \texttt{Timing} \cite{li2019time} (there is no source code available for \texttt{Hasse} \cite{sun2017hasse}). Since our algorithm uses the data structure \texttt{DCS} in \texttt{SymBi} \cite{min2021symmetric}, we implemented our algorithm by extending data structures (\texttt{DCS} and max-min timestamp) to incorporate temporal information and adding proposed techniques to \texttt{SymBi}.
Furthermore, we modified the state-of-the-art continuous subgraph matching algorithms (\texttt{SymBi} \cite{min2021symmetric} and \texttt{RapidFlow} \cite{sun2022rapidflow}) to solve our problem by additionally checking whether the embeddings found by \texttt{SymBi} and \texttt{RapidFlow} satisfy the temporal order, and included them in our comparisons.
Experiments are conducted on a CentOS machine with two Intel Xeon E5-2680 v3 2.50GHz CPUs and 256 GB memory. The source codes of \texttt{Timing}, \texttt{SymBi}, and \texttt{RapidFlow} were obtained from the authors and all methods are implemented in C++. Unlike the implementation of other comparative algorithms, \texttt{RapidFlow} assumes undirected edges in both data graph and query graph. We modified \texttt{RapidFlow} to output only the embeddings with correct directions.

\noindent\textbf{Datasets.} We use six datasets referred to as Netflow, Wiki-talk, Superuser, StackOverflow, Yahoo, and LSBench. Netflow is the anonymized network passive traffic data from CAIDA Internet Annoymized Traces 2015 Dataset \cite{caida2015}. Edges in Netflow are labeled with a triple of source ports, protocols, and destination ports.
Wiki-talk is a temporal network representing Wikipedia users editing each other's talk page from the Standford SNAP library \cite{snapstanford}. We label the vertices in Wiki-talk as the first character of the user's name as in \cite{li2019time}.
Superuser and StackOverflow are networks of user-to-user interactions on the stack exchange websites \cite{snapstanford}. The edges are labeled based on the type of interaction: (1) a user answering another user's question, (2) a user commenting on another user's question, and (3) a user commenting on another user's answer. 
Yahoo is a social network that captures the communication between users via Yahoo Messenger \cite{rossi2015network}. 
LSBench is synthetic social media stream data generated by the Linked Stream Benchmark data generator \cite{le2012linked}.  We set the number of users to 0.1 million and applied the default settings for other parameters.
For Superuser, StackOverflow, and Yahoo, we randomly label the vertices among five distinct labels so that a reasonable number of queries can finish within the time limit.
The characteristics of the datasets are summarized in Table \ref{tab:characteristics_datasets}, including the number of vertices $|V|$, the number of edges $|E|$, the number of distinct vertex labels $|\Sigma_V|$, the number of distinct edge labels $|\Sigma_E|$, the average degree $d_{avg}$, and the average number of parallel edges between a pair of adjacent vertices $m_{avg}$.

\addtolength{\tabcolsep}{-1pt}
\begin{table}[h]
    \centering
    \caption{Characteristics of datasets}
    \vspace*{-1mm}
    \begin{tabular}{ccccccc}
        \toprule
        Datasets & $|V|$ & $|E|$ & $|\Sigma_V|$ & $|\Sigma_E|$ & $d_{avg}$ & $m_{avg}$ \\ 
        \midrule
        Netflow     & 0.37M     & 15.96M    & 1     & 346,672    & 85.4     & 27.6     \\
        Wiki-talk   & 1.14M     & 7.83M     & 365   & 1         & 13.7     & 2.37      \\
        Superuser     & 0.19M    & 1.44M    & 5    & 3        & 14.9      & 1.56      \\ 
        StackOverflow     & 2.60M    & 63.50M    & 5    & 3        & 48.8      & 1.75      \\ 
        Yahoo     & 0.10M    & 3.18M    & 5    & 1        & 63.6      & 3.51      \\ 
        LSBench     & 13.12M    & 21.04M    & 11    & 19        & 3.21      & 1.00      \\
        \bottomrule
    \end{tabular}
    \vspace{-2mm}
    \label{tab:characteristics_datasets}
\end{table}
\addtolength{\tabcolsep}{1pt}

\noindent\textbf{Queries.} We adopt a query generation method in the previous study \cite{li2019time}. To generate query graphs, we first traverse the data graph by random walk. To ensure that there is a time-constrained embedding of the query graph in the data graph, the temporal order is determined as follows. We create a random permutation of edges in the query graph $q$ generated by random walk, and then for any two edges $e,e'\in E(q)$, we set $e\prec e'$ if (1) $e$ precedes $e'$ in the permutation and (2) the timestamp of $e$ is less than that of $e'$.

For each dataset, we set six different query sizes: 5, 7, 9, 11, 13, 15. The size of a query is defined as the number of edges. We generate 100 queries as mentioned above for each dataset and query size. For each query graph, we create 5 different temporal orders. The density of a temporal order is defined as the number of edge pairs with temporal relationships divided by the number of edge pairs in the query graph. To see the difference in performance according to the density of the temporal order, we generate two temporal orders with densities of 0 (empty) and 1 (total order), and the others to have densities close to 0.25, 0.50, and 0.75.

\noindent\textbf{Performance Measurement} We measure the elapsed time of time-constrained continuous subgraph matching for a dataset and a query graph. Since this problem is NP-hard, some queries may not finish in a reasonable time. Therefore, we set a time limit of 1 hour for each query. If an algorithm does not finish a query within the time limit, we regard the elapsed time of the query as 1 hour. We say that a query graph is \textit{solved} if it ends within the time limit. Each query set consists of 100 query graphs with a same size. For each query set, we measure the average elapsed time and the number of solved query graphs. For a reasonable comparison, we compute the average time excluding query graphs that all algorithms failed to solve. In addition, we measure the peak memory usage using the ``ps'' utility to compare the memory usage of programs.

\subsection{Experimental Results}
We evaluate the performance of the algorithms by varying the query size, the density, and the time window size. 
Because each dataset has a different time span, we set each unit of the window size as the average time span between two consecutive edges in the dataset. We used 5 different window sizes in our experiments: 10k, 20k, 30k, 40k, 50k. Table \ref{tab:experiment_settings} shows the parameters used in the experiments. Values in boldface are used as default parameters.

\begin{table}[h]
    \centering
    \caption{Experiment settings}
    \vspace*{-1mm}
    \begin{tabular}{cc}
    \toprule
    \textbf{Parameter} & \textbf{Value Used} \\ 
    \midrule
    Datasets & Netflow, Wiki-talk, Superuser, \\ 
        & StackOverflow, Yahoo, LSBench \\
    Query size & 5, 7, \textbf{9}, 11, 13, 15 \\
    Density & 0, 0.25 \textbf{0.50}, 0.75, 1 \\ 
    Window size & 10k, 20k, \textbf{30k}, 40k, 50k \\ 
    \bottomrule
    \end{tabular}
    \label{tab:experiment_settings}
\end{table}

\begin{figure*}[t]
    \centering
    \begin{subfigure}{\linewidth}
        \centering
        \includegraphics[scale=0.18]{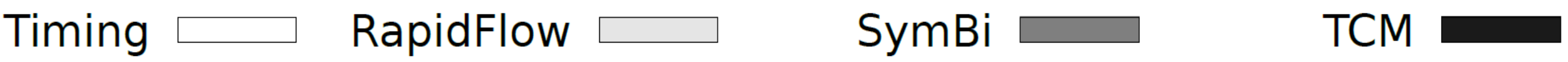}
    \end{subfigure}
    \subcaptionbox{Average elapsed time\label{exp:vary_query_size_time}}{
        \includegraphics[width=0.98\linewidth]{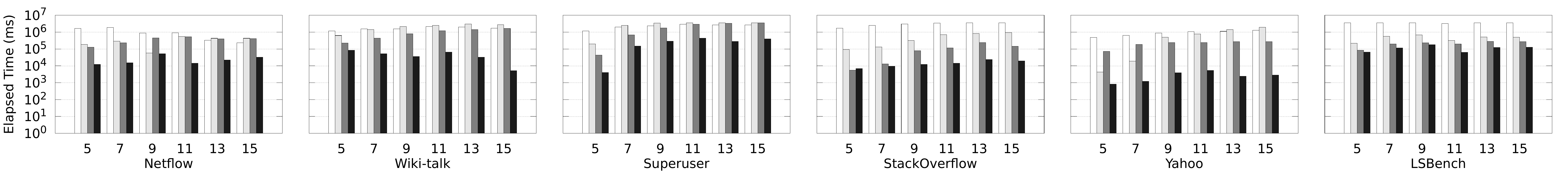}
        \vspace*{-1mm}
    }
    \subcaptionbox{Solved queries\label{exp:vary_query_size_solved}}{
        \includegraphics[width=0.98\linewidth]{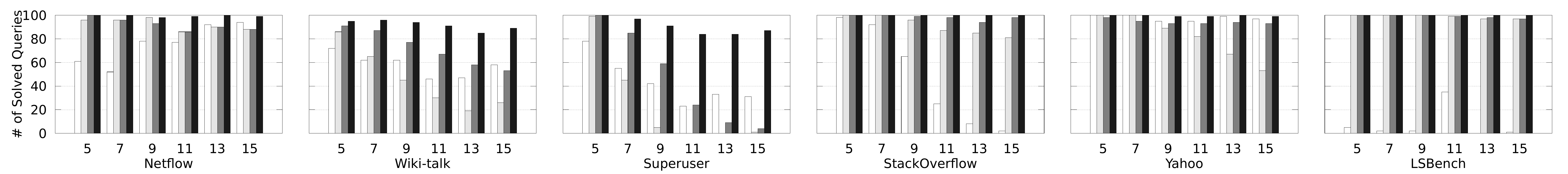}
        \vspace*{-1mm}
    }
    \vspace*{-2mm}
    \caption{Query processing time and the number of solved queries for varying query size}
    \label{exp:vary_query_size}
    \vspace*{-3mm}
\end{figure*}

\begin{figure*}[t]
    \centering
    \begin{subfigure}{\linewidth}
        \centering
        \includegraphics[scale=0.18]{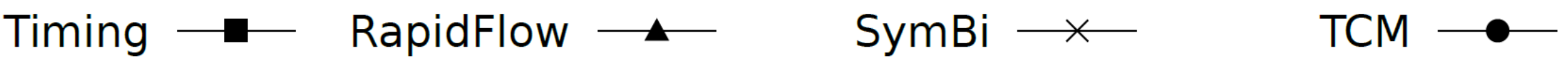}
    \end{subfigure}
    \subcaptionbox{Average elapsed time\label{exp:vary_tc_rate_time}}{
        \includegraphics[width=0.98\linewidth]{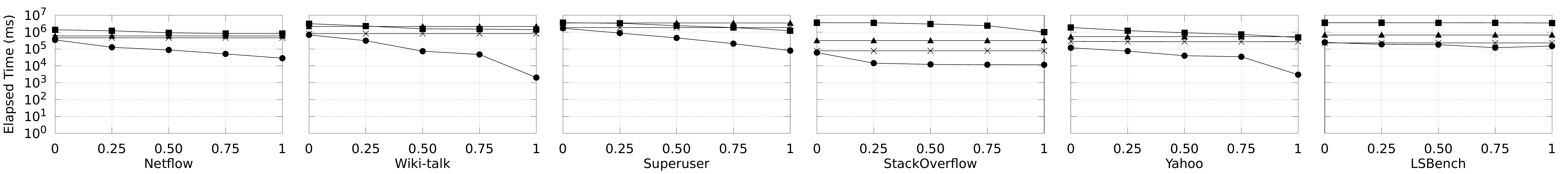}
        \vspace*{-1mm}
    }
    \subcaptionbox{Solved queries\label{exp:vary_tc_rate_solved}}{
        \includegraphics[width=0.98\linewidth]{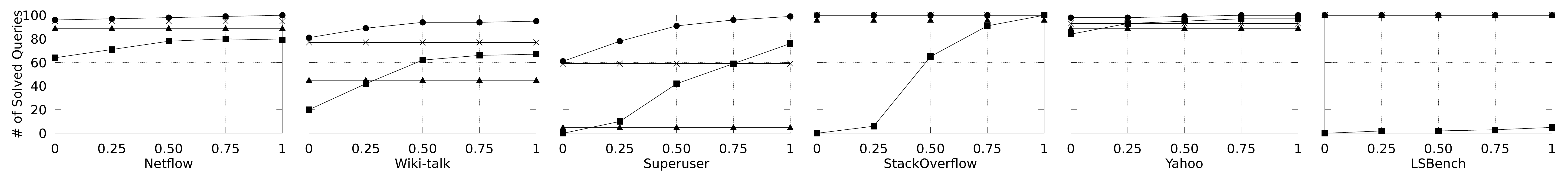}
        \vspace*{-1mm}
    }
    \vspace*{-2mm}
    \caption{Query processing time and the number of solved queries for varying density}
    \label{exp:vary_tc_rate}
    \vspace*{-3mm}
\end{figure*}

\noindent\textbf{Varying the query size.} 
Figure \ref{exp:vary_query_size} shows the performance results for varying the query size. We set the density to 0.50 and the window size to 30k. 

Our algorithm \texttt{TCM} outperforms the other algorithms in terms of both query processing time and the number of solved queries for all datasets and query sizes. Specifically, \texttt{TCM} is more than two orders of magnitude faster than \texttt{Timing} in terms of query processing time for all query sizes in both StackOverflow and Yahoo. Moreover, \texttt{TCM} is usually more than 10 times faster than \texttt{RapidFlow} and \texttt{SymBi} on most datasets and query sizes, and achieves more than 100 times speed up on large query sizes in Yahoo. With respect to the number of solved queries, in LSBench, \texttt{Timing} rarely solves queries except when the query size is 11, while \texttt{TCM} successfully solves all queries in all query sizes. Furthermore, \texttt{TCM} continues to solve a significant number of queries as the query size increases, whereas the number of queries solved by the other three algorithms drops sharply in most datasets. Each of the three comparing algorithms solves a large number of queries when the query size is small, but it solves less than 10 queries when the query size increases for some datasets.  These results suggest that \texttt{TCM} is a robust algorithm as it performs well across different datasets, while the other algorithms' performance varies significantly depending on the dataset.

In Wiki-talk, the performance gap between \texttt{TCM} and the other algorithms widens as the query size increases. As the query size increases, the query processing time of \texttt{TCM} decreases, while the query processing time of \texttt{Timing}, \texttt{RapidFlow}, and \texttt{SymBi} generally increases. This is because as the query size increases, the time constraints associated with one edge in the query graph increase, leading to better filtering and pruning in \texttt{TCM}. 
The following experiments with varying the density show that Wiki-talk is more affected by time constraints than other datasets.

\begin{figure*}[t]
    \centering
    \begin{subfigure}{\linewidth}
        \centering
        \includegraphics[scale=0.18]{figures/vary_tc_rate_key.pdf}
    \end{subfigure}
    \subcaptionbox{Average elapsed time\label{exp:vary_window_size_time}}{
        \includegraphics[width=0.98\linewidth]{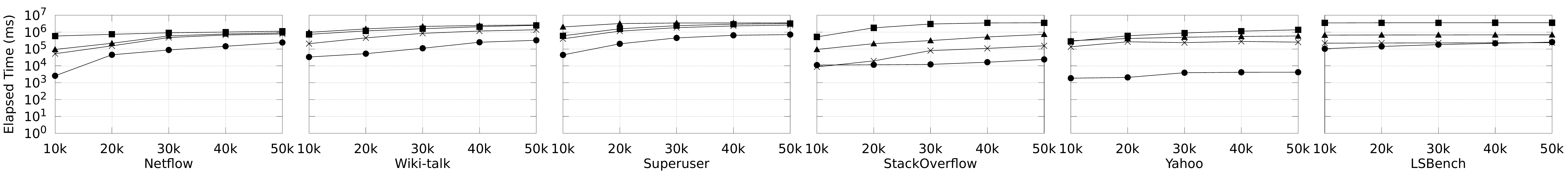}
        \vspace*{-1mm}
    }
    \subcaptionbox{Solved queries\label{exp:vary_window_size_solved}}{
        \includegraphics[width=0.98\linewidth]{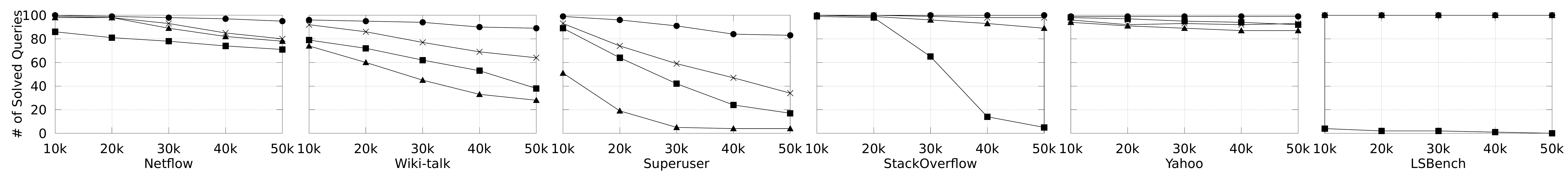}
        \vspace*{-1mm}
    }
    \vspace*{-2mm}
    \caption{Query processing time and the number of solved queries for varying window size}
    \label{exp:vary_window_size}
    \vspace*{-5mm}
\end{figure*}

\begin{figure*}[t]
    \centering
    \begin{subfigure}{\linewidth}
        \centering
        \includegraphics[scale=0.18]{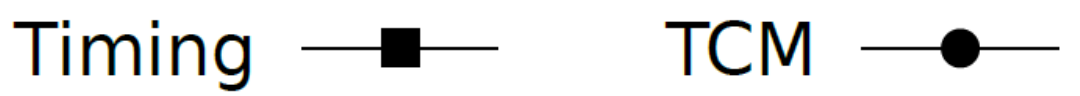}
    \end{subfigure}
    \begin{subfigure}{\linewidth}
        \centering
        \includegraphics[width=0.98\linewidth]{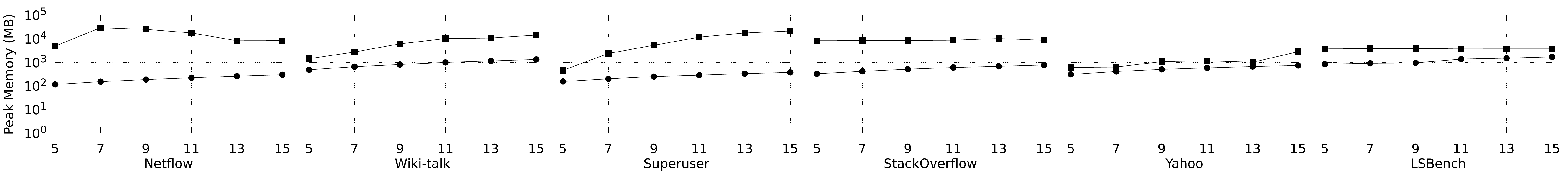}
    \end{subfigure}
    \vspace*{-4mm}
    \caption{Average peak memory for varying query size}
    \label{exp:memory_usage}
    \vspace*{-3mm}
\end{figure*}

\begin{figure*}[t]
    \centering
    \begin{subfigure}{\linewidth}
        \centering
        \includegraphics[scale=0.18]{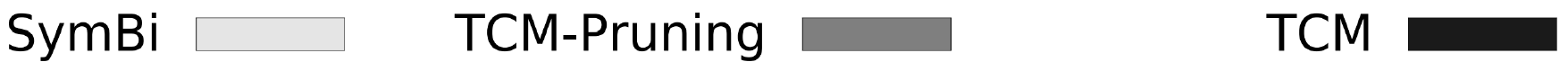}
    \end{subfigure}
    \subcaptionbox{Average elapsed time\label{exp:techniques_time}}{
        \includegraphics[width=0.98\linewidth]{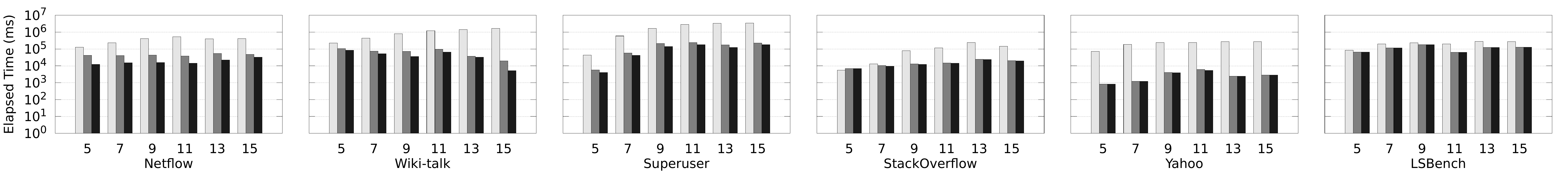}
        \vspace*{-1mm}
    }
    \subcaptionbox{Solved queries\label{exp:techniques_solved}}{
        \includegraphics[width=0.98\linewidth]{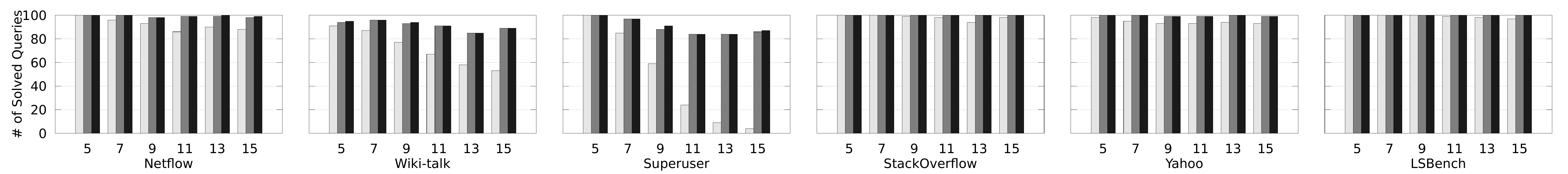}
        \vspace*{-1mm}
    }
    \vspace*{-2mm}
    \caption{Evaluating techniques for varying query size}
    \label{exp:techniques}
    \vspace*{-5mm}
\end{figure*}

\noindent\textbf{Varying the density.} Next, we fix the query size to 9 and the window size to 30k, and vary the density from 0 to 1 as shown in Figure \ref{exp:vary_tc_rate}. To see the effect of the density, we exclude queries that all algorithms failed to solve in every density when computing the average query processing time. 

Both \texttt{RapidFlow} and \texttt{SymBi} find embeddings without considering a temporal order to solve the problem, so the query processing time does not change even if the density changes. On the other hand, \texttt{Timing} and \texttt{TCM} generally decrease as the density increases. Figure \ref{exp:vary_tc_rate_time} shows that the query processing time of \texttt{TCM} is noticeably reduced as the density increases. In contrast, the query processing time of \texttt{Timing} does not decrease significantly like \texttt{TCM}. 
As a result, \texttt{TCM} outperforms \texttt{Timing} in all six datasets, with the performance gap increasing as the density increases. Specifically, when the density is 1, \texttt{TCM} is several hundred times faster than \texttt{Timing} in Wiki-talk and several times faster in the other datasets. Moreover, even when the density is 0, \texttt{TCM} still achieves a notable performance improvement compared to \texttt{Timing}.
These results show that \texttt{TCM} utilizes the temporal order better than \texttt{Timing} to solve the time-constrained continuous subgraph matching problem. Also, this supports that query processing time is more affected by time constraints in Wiki-talk than in other datasets.

\noindent\textbf{Varying the window size.} In this experiment, we examine the effect of the window size on the performance by varying it from 10k to 50k, while using a query size of 9 and density of 0.50. Figure \ref{exp:vary_window_size} represents the performance results. We can see that the query processing time increases and the number of solved queries decreases as the window size becomes larger. This is because as the window size increases we have to maintain more data edges and more time-constrained embeddings occur. 
In elapsed time, \texttt{TCM} maintains its superior performance in most cases, compared to other algorithms. 
As shown in Figure \ref{exp:vary_window_size_solved}, as the window size increases, the number of queries solved by \texttt{Timing}, \texttt{RapidFlow}, and \texttt{SymBi} decreases sharply, but \texttt{TCM} does not.

\noindent\textbf{Memory usage.} Figure \ref{exp:memory_usage} describes the peak memory within the time limit for varying the query size, which is averaged over 100 queries. We measure the peak memory as the maximum virtual set size in the ``ps'' utility output. 
In datasets other than Yahoo and LSBench, \texttt{TCM} uses 10 times less memory than \texttt{Timing}. In particular, there is a difference of more than 100 times in several query sizes of Netflow. This difference in memory usage comes from \texttt{TCM} using a data structure with polynomial space, while \texttt{Timing} stores all partial matches and thus requires exponential space. Moreover, as the query size increases, the gap in memory usage between \texttt{TCM} and \texttt{Timing} tends to widen.

\subsection{Effectiveness of Our Techniques}\label{subsec:effectiveness}

We evaluate the effectiveness of our techniques in this subsection. To measure the performance gains obtained by each technique, we implement a variant of our algorithm and compare it with \texttt{SymBi} and \texttt{TCM}.

\begin{itemize}
    \item \texttt{SymBi}: a baseline for comparison.
    \item \texttt{TCM-Pruning}: using TC-matchable edges and not using time-constrained pruning techniques.
    \item \texttt{TCM}: using all techniques.
\end{itemize}
In this evaluation, we vary the query size, and fix the density to 0.50 and the window size to 30k. Figure \ref{exp:techniques} shows the results.

\vspace{1mm}
\noindent\textbf{Effectiveness of the TC-matchable edge.} 
To see the difference in the use of the TC-matchable edge technique, we compare the two algorithms \texttt{SymBi} and \texttt{TCM-Pruning}. Figure \ref{exp:techniques_time} shows that \texttt{TCM-pruning} outperforms \texttt{SymBi} in terms of query processing time. Specifically, \texttt{TCM-Pruning} is more than 10 times faster than \texttt{SymBi} across multiple datasets. Especially, \texttt{TCM-Pruning} outperforms \texttt{SymBi} by 152.75 times in Yahoo when the query size is 7.
In terms of the number of solved queries, \texttt{TCM-Pruning} solves much more queries than \texttt{SymBi} in all datasets within the time limit.

\fullpaper{
To evaluate the effectiveness of the TC-matchable edge, we also compare the filtering power with and without the TC-matchable edge. We measure the filtering power with two factors. One is the number of edges in \texttt{DCS}, i.e., the number of edge pairs (pairs of edges in the query graph and edges in the data graph) that pass filtering when the TC-matchable edge is used, and the number of edge pairs with the same label when the TC-matchable edge is not used. The other is the number of vertices remaining in \texttt{DCS} after the filtering in \cite{min2021symmetric}. Table \ref{tab:filtering_power} shows the average values obtained by dividing the value measured when the TC-matchable edge is used by the value measured when not using the TC-matchable edge. A smaller value in the table means more filtering when using the TC-matchable edge. 
}

\fullpaper{
\begingroup
\setlength{\tabcolsep}{4.5pt} %
\begin{table}[h]
    \centering
    \caption{Filtering power with and without the TC-matchable edge. Top: the ratio of the number of edges in \texttt{DCS}, bottom: the ratio of the number of vertices remaining in \texttt{DCS} after the filtering in \cite{min2021symmetric}.}
    \vspace*{-1mm}
    \begin{tabular}{cccccccc}
    \toprule
    \textbf{} & \textbf{5} & \textbf{7} & \textbf{9} & \textbf{11} & \textbf{13} & \textbf{15} & \textbf{avg} \\ 
    \midrule
    Netflow & 0.286 & 0.234 & 0.255 & 0.265 & 0.260 & 0.227 & 0.254 \\
    Wiki-talk & 0.286 & 0.204 & 0.172 & 0.176 & 0.151 & 0.127 & 0.186 \\
    Superuser & 0.337 & 0.270 & 0.223 & 0.247 & 0.185 & 0.149 & 0.235 \\
    StackOverflow & 0.233 & 0.173 & 0.132 & 0.124 & 0.102 & 0.068 & 0.138 \\
    Yahoo & 0.245 & 0.172 & 0.117 & 0.112 & 0.072 & 0.049 & 0.128 \\
    LSBench & 0.608 & 0.609 & 0.557 & 0.503 & 0.443 & 0.455 & 0.529 \\
    \midrule
    Netflow & 0.705 & 0.684 & 0.501 & 0.393 & 0.262 & 0.139 & 0.447 \\
    Wiki-talk & 0.499 & 0.306 & 0.221 & 0.208 & 0.199 & 0.131 & 0.261 \\
    Superuser & 0.619 & 0.497 & 0.427 & 0.401 & 0.313 & 0.234 & 0.415 \\
    StackOverflow & 0.311 & 0.204 & 0.159 & 0.100 & 0.067 & 0.077 & 0.153 \\
    Yahoo & 0.338 & 0.198 & 0.111 & 0.082 & 0.036 & 0.032 & 0.133 \\
    LSBench & 0.415 & 0.423 & 0.319 & 0.233 & 0.255 & 0.209 & 0.309 \\
    \bottomrule
    \end{tabular}

    \label{tab:filtering_power}
\end{table}
\endgroup
}

\vspace{1mm}
\noindent\textbf{Effectiveness of time-constrained pruning techniques.} 
When comparing the query processing time of \texttt{TCM-Pruning} and \texttt{TCM}, our time-constrained pruning techniques showed significant improvements in performance across various datasets (Figure \ref{exp:techniques_time}). Specifically, the average improvement in query processing time is 2.60 times in Netflow, 1.83 times in Wiki-talk, 1.40 times in Superuser, 1.04 times in StackOverflow, 1.03 times in Yahoo, and 1.00 times in LSBench.
Figure \ref{exp:techniques_solved} shows that \texttt{TCM} using time-constrained pruning techniques solves the same or more queries within the time limit than \texttt{TCM-Pruning} without the techniques.

\section{Conclusion}\label{sec:conclusion}
In this paper, we proposed two key techniques for time-constrained continuous subgraph matching to address the existing limitations: (1) TC-matchable edge for filtering and (2) a set of pruning techniques in backtracking. The former allows us to utilize temporal relationships between non-incident edges in filtering, and the latter allows us to prune some of the edges in the backtracking. We also suggest a data structure called max-min timestamp for the TC-matchable edge with polynomial space and an efficient method to update the data structure. Extensive experiments on real and synthetic datasets show that our approach outperforms the state-of-the-art algorithm in terms of query processing time and the number of queries solved within the time limit. Parallelizing our approach is an interesting future work.

\section*{Acknowledgements}\label{sec:ack}
S. Min, J. Jang, and K. Park were supported by Institute of Information communications Technology Planning Evaluation (IITP) grant funded by the Korea government (MSIT) (No. 2018-0-00551, Framework of Practical Algorithms for NP-hard Graph Problems). Giuseppe F. Italiano was partially supported by MUR. the Italian Ministry of University and Research, under PRIN Project n. 2022TS4Y3N - EXPAND: scalable algorithms for EXPloratory Analyses of heterogeneous and dynamic Networked Data. W.-S. Han was supported by the National Research Foundation of Korea (NRF) grant funded by the Korea government (MSIT) (No. NRF-2021R1A2B5B03001551).

\balance
\bibliographystyle{IEEEtran}
\bibliography{IEEEabrv,ref}

\vspace{12pt}

\end{document}